\newif\ifDEBUG
\DEBUGfalse

\newif\ifANONYMOUS
\ANONYMOUSfalse

\newif\ifARXIV
\ARXIVtrue
\documentclass[letterpaper,twocolumn,10pt]{article}
\usepackage{usenix}

\usepackage[export]{adjustbox} % Aligning differently-sized subfigures
\usepackage{amsmath}
\usepackage{amsfonts}
\usepackage{array}
\usepackage[USenglish]{babel}
\usepackage{packages/bootstrapicons}
\usepackage{booktabs}
\usepackage{colortbl}
\usepackage{hyperref}
\usepackage{cleveref}
\usepackage{enumitem}
\usepackage{float}
\usepackage{forest}
\usepackage[htt]{hyphenat}
\usepackage{makecell}
\usepackage{multirow}
\usepackage{textgreek}
\usepackage{soul}
\usepackage{pifont}
\usepackage{subcaption} % Subfigures etc.
\usepackage{url}
\usepackage[table]{xcolor}
\usepackage{xspace}
\usepackage{wrapfig}
\usepackage{algorithm}
\usepackage{algorithmic}
\makeatletter
\newcommand{\FUNCTION}[1]{\ALC@it\textbf{function} #1\begin{ALC@g}}
\newcommand{\ENDFUNCTION}{\end{ALC@g}\ALC@it\textbf{end function}}
\newcommand{\SWITCH}[1]{\ALC@it\textbf{switch} #1 \textbf{do}\begin{ALC@g}}
\newcommand{\ENDSWITCH}{\end{ALC@g}}
\newcommand{\CASE}[1]{\ALC@it\textbf{case} #1\begin{ALC@g}}
\newcommand{\ENDCASE}{\end{ALC@g}}
\makeatother
\usepackage{paracol}
\usepackage{tabularx}
\usepackage{threeparttable} % For table notes
\usepackage{packages/picins}

\usepackage{tikz}
\usetikzlibrary{decorations.pathmorphing}
\usetikzlibrary{automata} % Import library for drawing automata
\usetikzlibrary{positioning} % ...positioning nodes
\usetikzlibrary{arrows} % ...customizing arrows
\usetikzlibrary{calc} % ...fancy edge calculations
\usetikzlibrary{patterns} % ...shading
\usetikzlibrary{snakes} % ...snakes
\usetikzlibrary{shapes} % ...shapes
\usetikzlibrary{arrows.meta} % ...arrows

\newcommand{\ie}{\textit{i.e.,\ }}
\newcommand{\eg}{\textit{e.g.,\ }}
\newcommand{\etal}{\textit{et al.\ }}

\usepackage{mdframed}
\mdfsetup{skipabove=0.5\topskip,skipbelow=0.5\topskip,align=center}

\usepackage{amsthm}
\usepackage{packages/coloredtheorem}
\tcbuselibrary{theorems}

\newtcbtheorem{definition}{Definition}{
  breakable,
  enhanced,
  colframe=blue!50!black,
  colback=white,
  colbacktitle=blue!50!black,
  coltitle=white,
  fonttitle=\bfseries,
  boxrule=0.5pt,
  arc=2pt,
  left=4pt, right=4pt, top=4pt, bottom=4pt
}{def}

\newtcbtheorem{theorem}{Theorem}{
  breakable,
  enhanced,
  colframe=green!50!black,
  colback=white,
  colbacktitle=green!50!black,
  coltitle=white,
  fonttitle=\bfseries,
  boxrule=0.5pt,
  arc=2pt,
  left=4pt, right=4pt, top=4pt, bottom=4pt
}{thm}

\newtcbtheorem{lemma}{Lemma}{
  breakable,
  enhanced,
  colframe=yellow!60!black,
  colback=white,
  colbacktitle=yellow!60!black,
  coltitle=white,
  fonttitle=\bfseries,
  boxrule=0.5pt,
  arc=2pt,
  left=4pt, right=4pt, top=4pt, bottom=4pt
}{lem}

\newtcbtheorem{example}{Example}{
  breakable,
  enhanced,
  colframe=gray!60,
  colback=white,
  colbacktitle=gray!60,
  coltitle=white,
  fonttitle=\bfseries,
  boxrule=0.5pt,
  arc=2pt,
  left=4pt, right=4pt, top=4pt, bottom=4pt
}{ex}

\newtcbtheorem{exploit}{Exploit}{
  breakable,
  enhanced,
  colframe=gray!60,
  colback=white,
  colbacktitle=gray!60,
  coltitle=white,
  fonttitle=\bfseries,
  boxrule=0.5pt,
  arc=2pt,
  left=4pt, right=4pt, top=4pt, bottom=4pt
}{exp}

\newtcbtheorem{pattern}{Pattern}{
  breakable,
  enhanced,
  colframe=purple!50!black,
  colback=white,
  colbacktitle=purple!50!black,
  coltitle=white,
  fonttitle=\bfseries,
  boxrule=0.5pt,
  arc=2pt,
  left=4pt, right=4pt, top=4pt, bottom=4pt
}{pat}

\newtcolorbox{rqanswer}[1][]{
  breakable,
  enhanced,
  colframe=teal!70!black,
  colback=white,
  colbacktitle=teal!70!black,
  coltitle=white,
  fonttitle=\bfseries,
  title={Answer to #1},
  boxrule=0.5pt,
  arc=2pt,
  left=4pt, right=4pt, top=4pt, bottom=4pt
}

\newtheorem{oldtheorem}{Theorem}

\crefname{tcb@cnt@definition}{Definition}{Definitions}
\crefname{tcb@cnt@theorem}{Theorem}{Theorems}
\crefname{tcb@cnt@lemma}{Lemma}{Lemmata}
\crefname{tcb@cnt@example}{Example}{Examples}
\crefname{tcb@cnt@pattern}{Pattern}{Patterns}
\Crefname{tcb@cnt@definition}{Definition}{Definitions}
\Crefname{tcb@cnt@theorem}{Theorem}{Theorems}
\Crefname{tcb@cnt@lemma}{Lemma}{Lemmata}
\Crefname{tcb@cnt@example}{Example}{Examples}
\Crefname{tcb@cnt@pattern}{Pattern}{Patterns}
\Crefname{tcb@cnt@exploit}{Exploit}{Exploits}

\newcommand*\circled[1]{\tikz[baseline=(char.base)]{
            \node[shape=circle,draw,inner sep=0.5pt] (char) {#1};}}

\newlist{researchquestions}{enumerate}{2}
\setlist[researchquestions,1]{
  label=\textbf{RQ\arabic*},
  leftmargin=*,
  ref=RQ\arabic*
}

\setlist[researchquestions,2]{
  label=\textbf{RQ\arabic{researchquestionsi}\alph*},
  ref=RQ\arabic{researchquestionsi}\alph*,
  leftmargin=2em
}

\newlist{exploitdesc}{description}{10}
\setlist[exploitdesc]{
    leftmargin=0pt,
    itemsep=0pt
}

\crefformat{section}{{\S}#2#1#3}%
\crefname{figure}{Figure}{Figures}%
\crefname{appendix}{Appendix}{Appendices}%
\crefname{table}{Table}{Tables}%
\crefname{algorithm}{Algorithm}{Algorithms}%
\crefname{listing}{Listing}{Listings}%
\crefname{theorem}{Theorem}{Theorems}%
\crefname{oldtheorem}{Theorem}{Theorems}%
\crefname{thm}{Theorem}{Theorems}%
\crefname{lemma}{Lemma}{Lemmata}%
\crefname{equation}{Eqt.}{Eqts.}%
\crefformat{Grammar}{Grammar #1}%
\crefname{definition}{Definition}{Definitions}%
\crefname{example}{Example}{Examples}%

\Crefformat{section}{{\S}#2#1#3}%
\Crefname{figure}{Figure}{Figures}%
\Crefname{appendix}{Appendix}{Appendices}%
\Crefname{table}{Table}{Tables}%
\Crefname{algorithm}{Algorithm}{Algorithms}%
\Crefname{listing}{Listing}{Listings}%
\Crefname{thm}{Theorem}{Theorems}%
\Crefname{lemma}{Lemma}{Lemmata}%
\Crefname{equation}{Eqt.}{Eqts.}%
\Crefformat{Grammar}{Grammar #1}%
\Crefname{definition}{Definition}{Definitions}%

\newcommand{\myparagraph}[1]{\paragraph{#1}}
\ifDEBUG
  \PassOptionsToPackage{draft}{commenting}
\else
  \PassOptionsToPackage{final}{commenting}
\fi

\usepackage[nompar]{packages/commenting}

\declareauthor{JD}{Jamie}{purple}
\declareauthor{BC}{Berk}{green}
\declareauthor{YL}{Yichen}{blue}
\declareauthor{MK}{Minseok}{red}
\declareauthor{AM}{Aman}{cyan}
\declareauthor{DY}{Dongyoon}{orange}

\authorcommand{JD}{comment}
\authorcommand{BC}{comment}
\authorcommand{YL}{comment}
\authorcommand{MK}{comment}
\authorcommand{AM}{comment}
\authorcommand{DY}{comment}

\onlyauthors{JD,BC,YL,MK,AM,DY}

\usepackage{amssymb}
\usepackage{bbm}
\DeclareSymbolFont{bbm}{U}{bbm}{m}{n}
\DeclareMathSymbol{\lcg}{\mathord}{bbm}{"28}
\DeclareMathSymbol{\rcg}{\mathord}{bbm}{"29}
\usepackage{mathtools}
\renewcommand{\land}{\mathrel{\wedge}}
\renewcommand{\lor}{\mathrel{\vee}}
\renewcommand{\emptyset}{\varnothing}

\newcommand{\algohline}{\vspace{-1em}\rule{\linewidth}{0.5pt}}
\newcommand{\ALGOSIG}{\item[\textbf{Algorithm}]}

\usepackage{caption}
\usepackage{cprotect}
\usepackage{listings}

\lstnewenvironment{exploitlst}{\lstset{
    basicstyle=\small\ttfamily,
    aboveskip=0pt,
    belowskip=0pt,
    breaklines,
    breakindent=0ex
}}{}

\usepackage{titlesec}
\titlespacing*{\paragraph}{0pt}{2ex plus 1ex minus .2ex}{1em}

\begin{document}
%-------------------------------------------------------------------------------

%don't want date printed
\date{}

\title{\Large \bf Regular Expression Denial of Service Induced by Backreferences}

%for single author (just remove % characters)
\ifANONYMOUS
    \author{Anonymous Author(s)}
\else
\author{
{\rm Yichen Liu}\\
Stony Brook University
\and
{\rm Berk \c{C}akar}\\
Purdue University
\and
{\rm Aman Agrawal\thanks{Contribution made at Stony Brook University; affiliated with Google at the time of submission.}}\\
Stony Brook University
\and
{\rm Minseok Seo\thanks{Contribution made at Stony Brook University.}} \\ % unaffiliated at the time of submission.
Stony Brook University
\and
{\rm James C. Davis}\\
Purdue University
\and
{\rm Dongyoon Lee}\\
Stony Brook University
% copy the following lines to add more authors
% \and
% {\rm Name}\\
%Name Institution
} % end author
\fi

\maketitle

%-------------------------------------------------------------------------------
\begin{abstract}
%-------------------------------------------------------------------------------

This paper presents the first systematic study of denial-of-service vulnerabilities in Regular Expressions with Backreferences (REwB). We introduce the Two-Phase Memory Automaton (2PMFA), an automaton model that precisely captures REwB semantics. Using this model, we derive necessary conditions under which backreferences induce super-linear backtracking runtime, even when sink ambiguity is linear---a regime where existing detectors report no vulnerability. Based on these conditions, we identify three vulnerability patterns, develop detection and attack-construction algorithms, and validate them in practice.
Using the Snort intrusion detection ruleset, our evaluation identifies 45 previously unknown REwB vulnerabilities with quadratic or worse runtime. We further demonstrate practical exploits against Snort, including slowing rule evaluation by 0.6-1.2 seconds and bypassing alerts by triggering PCRE’s matching limit.

\end{abstract}

%-------------------------------------------------------------------------------
\section{Introduction}
%-------------------------------------------------------------------------------

% Paragraph 1: Problem and Context
Regular expressions (regexes) are a foundational mechanism for pattern matching and input validation across software systems.
They are widely used to validate and filter untrusted input, including network intrusion detection systems~\cite{roeschSnortLightweightIntrusion1999,RegexSnortRule}, web application firewalls~\cite{ModSecurity,OWASPCoreRuleSet,WafRegexPatternSet}, and server-side input validators~\cite{Chapman2016RegexUsageInPythonApps,barlas_exploiting_2022}.
However, many modern regex engines employ backtracking-based matching algorithms that can exhibit super-linear time complexity on certain regex pattern and input pairs~\cite{SoKReDoS}.

This algorithmic complexity vulnerability, known as Regular Expression Denial of Service (ReDoS)~\cite{Crosby2003REDOS,Crosby2003AlgorithmicComplexityAttacks,SoKReDoS}, has caused significant real-world outages,
    including a 34-minute downtime of Stack Overflow in 2016~\cite{StackOverflow2016REDOSPostMortem}
    and a 27-minute global outage of Cloudflare services in 2019~\cite{Cloudflare2019REDOSPostMortem}.
The prevalence of ReDoS vulnerabilities across software ecosystems is well-established through prior empirical studies (\eg ~\cite{davis2018,liuRevealerDetectingExploiting2021,wustholz2017,li_redoshunter_2021,Staicu2018REDOS,Shen2018ReScueGeneticRegexChecker,mclaughlin_regulator_2022}), which have collectively identified hundreds of vulnerable regex patterns in production code.

% Paragraph 2: Gap in Knowledge
While prior work provides evidence that ReDoS vulnerabilities are widespread,
  the existing theoretical basis for ReDoS focuses on Kleene regexes (K-regexes)---regexes constructed using only concatenation, alternation, and repetition operators---and their corresponding Non-deterministic Finite Automata (NFAs)~~\cite{Allauzen2008TestingAmbiguity,weber1991ida,weideman2016sink,wustholz2017}.
%\JD{Next sentence should be moved to \$2. After you do so, you need to cut some other phrases later in the intro (reference to sink ambiguity).}
%These approaches demonstrate that super-linear time complexity occurs when a regex’s corresponding NFA has an superlinear degree of sink ambiguity,
%and this automaton-level analysis has led to the development of detection tools for identifying ReDoS vulnerable K-regex patterns.
Yet, modern regex engines, such as those used in Python, Perl, PHP, and Java, commonly support backreferences and other extended constructs~\cite{SoKReDoS,Becchi2008CountingAutomata}, which cannot be represented by NFAs and therefore fall outside the scope of existing NFA-based complexity analyses.
Prior work has examined the expressive power of regexes with backreferences (REwB)~\cite{nogamiRegularExpressionsBackreferences2024,nogamiExpressivePowerRegular2023,Campeanu2003FormalPracticalRegexes,berglundReexaminingRegularExpressions2023}, and it is known that regex matching with backreferences is NP-complete~\cite{ahoPatternMatchingStrings1980}. However, this worst-case complexity result does not characterize which specific REwB patterns lead to super-linear backtracking behavior, nor does it provide practical algorithms for detecting such patterns or constructing attack inputs.
% \JD{I think the next sentence is not helpful, and not needed. Suggest we remove it.}
% At the same time, a survey on regex engineering practices emphasizes that developers struggle with regex complexity and often unknowingly introduce high-risk patterns~\cite{Michael2019RegexesAreHard}.
The prevalence of such patterns in real-world deployments also remains unknown.

% Paragraph 3: Goal/Method
%\BC{Should we add 2PMFA etc. details there?}
%\JD{The next paragraph feels repetitive with the abstract. Prefer to put more detail here and less in the abstract.}
%To address this gap, we extend ReDoS theory and detection to support REwB, providing the first systematic investigation of ReDoS vulnerabilities caused by REwB.
%First, we formalize three classes of REwB patterns capable of inducing super-linear backtracking behavior, even when the regex's sink ambiguity is linear.
%Then, we develop detection algorithms that identify these patterns by analyzing the structure of Two-Phase Memory Automata (2PMFA), a new automaton model we propose for REwB with self-references.
%For each detected pattern, our algorithms construct an attack automaton from which adversarial inputs can be extracted.
%Together, our contributions enable the identification of REwB-related ReDoS vulnerabilities. %that were previously invisible to existing detectors.

To address this gap, we extend ReDoS theory and detection to support REwB, providing the first systematic investigation of ReDoS vulnerabilities caused by backreferences.
We introduce Two-Phase Memory Automata (2PMFA), a new automaton model that faithfully captures real-world REwB semantics, including self-references.
Using 2PMFA, we formally show that for certain REwB patterns, a single backreference evaluation can incur non-$\mathcal{O}(1)$ cost. When combined with a non-$\mathcal{O}(1)$ number of backreference evaluations, this leads to super-linear runtime behavior that fundamentally differs from that of K-regexes. 
%In K-regex analysis, all transitions have $\mathcal{O}(1)$ cost, and super-linear runtime arises only from double-overlapping loops~\cite{weber1991ida,alluazen2008}.

Building on these insights, we formally derive necessary conditions under which REwB induce super-linear backtracking runtime due to non-$\mathcal{O}(1)$ per-backreference cost and a non-$\mathcal{O}(1)$ number of backreference evaluations. By combining these conditions, we introduce three ReDoS-vulnerable REwB patterns for the first time. Based on the patterns, we develop a ReDoS detector for REwB as well as attack-automaton generators.
Collectively, these contributions enable the identification of REwB-related ReDoS vulnerabilities that were previously invisible to existing detectors.

% Paragraph 4: Evaluation/Results
\JD{Berk, please use the structure from my recent email to break up this Eval section as `two parts'}
We evaluate our detection framework on 11K+ Snort~\cite{roeschSnortLightweightIntrusion1999} intrusion detection rules, and uncover 45 previously undocumented REwB-induced ReDoS vulnerabilities.
%\JD{Work with Dr. Lee about our evaluation. We must do Snort. Would be nice to do ClamAV as well which has backreferences, I think, and to include *case studies* demonstrating exploits based on behavior of Snort and ClamAV. Berk has suggested we also consider existing ReDoS defenses in several SOTA engines, which I think is a good idea.}
%In particular, our detector successfully identified
%    XXX vulnerable REwB patterns,
%    including
%        Y instances of Class 1,
%        Z instances of Class 2...
%Detection completed within 1.5 seconds per regex for over XX\% of patterns.
Through dynamic analysis, we validate our detector's findings and demonstrate that exploiting backreferences in combination with infinite degree of ambiguity (IDA) patterns produces substantially larger slowdowns and enables ReDoS attacks with shorter inputs compared to exploiting IDA alone.
Finally, we present four concrete exploits against Snort, together with malicious input strings and realistic attack scenarios, that either slow rule evaluation by 0.6–1.2 seconds or bypass alerts by triggering PCRE's matching limit.

%Overall, our results confirm that backreferences represent a significant and previously uncharacterized attack surface for ReDoS vulnerabilities in software systems.

%\noindent 
In summary, this paper makes the following contributions:
%\begin{itemize}[itemsep=0.1em, topsep=0.1em]
\begin{itemize}[nosep]
    \item We introduce the Two-Phase Memory Finite Automaton (2PMFA), a new automaton model that captures REwB semantics, and enables formal complexity analysis of backreference-induced backtracking behavior (\cref{sec:two_phase_memory_automaton}).
    \item We prove necessary conditions under which REwB incur super-linear time complexity in a manner that fundamentally differs from K-regexes (\cref{sec:theory}).
    %even if its sink ambiguity is linear (\cref{sec:theory}). %\YL{Remove ``We proved necessary and sufficient conditions under which backreferences cause super-linear time complexity''}.
    \item To the best of our knowledge, this is the first work to characterize backreference-induced ReDoS patterns, each of which is sufficient to induce super-linear time complexity. We develop a ReDoS-vulnerable REwB detector and an attack-automaton generator (\cref{sec:vulnerable_rewb_patterns}).
    %, even when the sink ambiguity is linear (\cref{sec:detection}).
    %\item We develop a detector to identify vulnerable REwB patterns in real-world regexes and provide automaton-based construction methods for generating attack inputs (\cref{subsec:attacking}).
    \item Our evaluation on Snort's intrusion detection rules uncovers previously unknown 45 REwB-induced ReDoS vulnerabilities and demonstrates realistic exploit scenarios against Snort, highlighting our findings' real-world impact (\cref{sec:evaluation}).
    %We evaluate our detector on Snort's intrusion detection ruleset, identifying XXX vulnerable REwB patterns among YYY unique regexes, and demonstrate through dynamic analysis that combining backreferences with IDA patterns enables more dangerous ReDoS attacks (\cref{sec:evaluation}).
\end{itemize}

% \JD{Berk, sketch a `Significance' paragraph here. It forces us to reinterpret our results, and it also helps the reader shift their point of view to see why the work matters, before they keep reading. The first part of the intro + the contrib bullets, these say WHAT we did. The Significance paragraph, this says WHY we did and what its POTENTIAL is.}
\smallskip
\noindent\textbf{\ul{Significance:}}
Our study is the first to systematically analyze ReDoS vulnerabilities induced by backreferences.
The findings reveal that backreferences
introduce a fundamentally distinct source of super-linear runtime that existing NFA-based detectors cannot capture: a single loop interacting with a backreference suffices to cause super-linear behavior. %, even when sink ambiguity is linear.
We recommend that developers and
operators of security-critical regex deployments carefully audit REwB using our patterns, rather than relying solely on previous tooling that can miss such vulnerabilities. 
Researchers % extending ReDoS theory 
should adopt automaton models that account for non-$\mathcal{O}(1)$ cost transitions, such as our 2PMFA, as a foundation for analyzing other irregular regex constructs. % including lookaheads and atomic groups.

\iffalse
Regex

ReDoS

Regular Expression with Backreference (REwB)

Theorems Super-linear REwB

Super-linear REwB patterns and detection

Evaluation. Snort dataset.

Contributions
\fi

%-------------------------------------------------------------------------------
\section{Background}
\label{sec:background}
%-------------------------------------------------------------------------------

This section introduces the regex constructs central to our work (\cref{sec:background:regex}), then reviews the algorithmic and complexity foundations of ReDoS that we later extend (\cref{subsec:background_redos_and_regex_complexity}).

\subsection{Regular Expressions and Backreferences}
\label{sec:background:regex}

A regular expression (regex) $r$ formally describes a language---a set of strings over an alphabet $\varSigma$---through concatenation ($rr$), alternation ($r|r$), and repetition ($r^*$), with parentheses $(r)$ for grouping~\cite{Kleene1951NerveNetsAndRegularLanguages}. 
For example, the regex \texttt{/a(b*)c/} matches \texttt{`abbbc'} but rejects \texttt{`aabbc'}.
Regexes constructed solely from these operations are called \textit{Kleene regexes (K-regexes)}.

\begin{definition}{K-regexes and REwB}{regex_rewb}
The syntax of K-regexes over an alphabet $\varSigma$ is given by the six constructs below.
REwB are obtained by extending this syntax with the two highlighted constructs: capturing groups and backreferences.
\[
\renewcommand{\arraystretch}{1.1}
\begin{array}{@{}l@{\hspace{4px}}l@{\hspace{4px}}lll@{\hspace{4px}}ll@{}}
    r &::={} & rr              & \textit{concatenation}
      &\mid  & r^*             & \textit{repetition}  \\
      &\mid  & r \mid r        & \textit{alternation}
      &\mid  & (r)             & \textit{grouping} \\
      &\mid  & \sigma          & \textit{symbol}~(\sigma \in \varSigma)
      &\mid  & \epsilon        & \textit{empty string} \\
\rowcolor{blue!15}
      &\mid  & \lcg_i r \rcg_i & \textit{capturing group}
      &\mid  & \backslash i    & \textit{backreference}
\end{array}
\]
\end{definition}

\paragraph{Backreferences.}
Modern regexes extend K-regexes to \textit{Extended regexes (E-regexes)}, adding syntactic sugar (\eg one-or-more-repetition $r+$) and constructs such as backreferences and lookarounds.
We focus on backreferences, which give regexes a form of memory.
By \textit{Regexes with Backreferences (REwB)} we refer to K-regexes plus backreference semantics.

In REwB, a capturing group $\lcg_i r \rcg_i$ records the substring matched by $r$, and a subsequent backreference $\backslash i$ matches that substring.
For example, the regex \verb|(a*)b\1| captures a sequence of \texttt{`a'}s in group~1 via \verb|(a*)|, then requires the backreference \verb|\1| to match the identical sequence.
This regex accepts \texttt{`aaabaaa'} (where \texttt{`aaa'} is captured and repeated) but rejects \texttt{`aaabaa'} (captured \texttt{`aaa'} $\neq$ trailing \texttt{`aa'}).
A special case are \textit{self-backreferences}~\cite{berglund2017}, in which $\backslash i$ occurs within its own capturing group $\lcg_i \cdots \rcg_i$; at each iteration the backreference matches the substring captured in the \textit{preceding} iteration (\cref{fig:self_bref}; see \cref{app:self_bref} for a detailed walkthrough).

\begin{figure}[!ht]
\centering
\includegraphics[width=\linewidth]{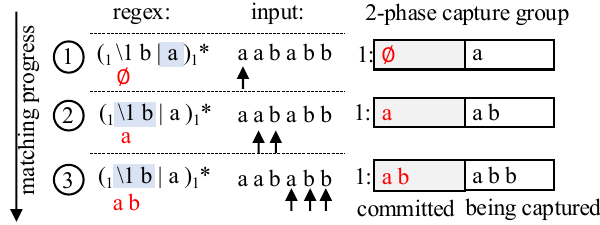}
\caption{Matching $\lcg_1 \backslash 1 \texttt{b} \mid \texttt{a} \rcg_1^*$ against \texttt{`aababb'}. The capture table stores the \textit{committed} value from the prior iteration alongside the substring \textit{being captured} in the current iteration.}
\label{fig:self_bref}
\end{figure}

\paragraph{Automata Equivalence and Irregular Constructs.}
K-regexes are equivalent in expressive power to regular languages: every K-regex can be converted to a non-deterministic finite automaton (NFA) via the Thompson-McNaughton-Yamada construction~\cite{Thompson1968LinearRegexAlgorithm,McNaughton1960Regexes}, and vice versa.
An NFA is a 5-tuple $A = (Q, \varSigma, \varDelta, q_0, F)$, where $Q$ is a finite set of states, $\varSigma$ is the input alphabet, $\varDelta: Q \times (\varSigma \cup \{\epsilon\}) \rightarrow \mathcal{P}(Q)$ is the transition function, $q_0 \in Q$ is the initial state, and $F \subseteq Q$ is the set of accepting states.
\Cref{fig:nfa}(1) shows the NFA for the regex \verb|a(b)*c|: starting from $q_0$, the automaton consumes \texttt{`a'}, loops on \texttt{`b'}, then accepts after \texttt{`c'}.

Some E-regex extensions---such as bounded quantifiers (\eg \verb|a{2,5}|) and character classes (\eg \verb|[a-z]|))---can be desugared into equivalent K-regex constructs and remain within the regular language class.
Backreferences, however, increase expressive power beyond regular languages: matching $\backslash i$ requires comparing the current input against a previously captured substring of arbitrary length, a dependency that memoryless NFAs cannot express.
Constructs that exceed regular expressiveness are termed \textit{irregular}~\cite{Becchi2008CountingAutomata,Campeanu2003FormalPracticalRegexes}; patterns containing them cannot be modeled by an NFA.

\begin{figure}[t]
  \centering
  \includegraphics[width=\linewidth]{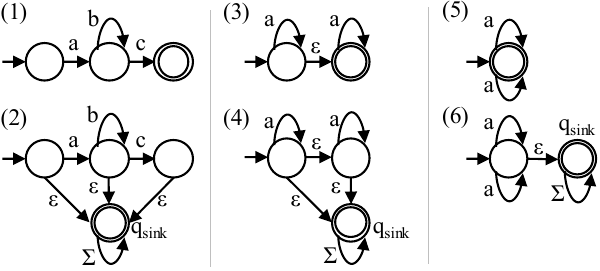}
  \caption{(1) NFA and (2) Sink-NFA of regex \texttt{/a(b*)c/}. (3) NFA and (4) Sink-NFA of regex \texttt{/a*b*/}. (5) NFA and (6) Sink-NFA of regex \texttt{/(a|a)*/}.}
  \label{fig:nfa}
\end{figure}

\subsection{ReDoS and Regex Complexity}
\label{subsec:background_redos_and_regex_complexity}

Regular Expression Denial of Service (ReDoS)~\cite{Crosby2003REDOS,Crosby2003AlgorithmicComplexityAttacks,SoKReDoS} is an algorithmic complexity attack in which a crafted input triggers worst-case behavior---polynomial or exponential in input length---in a backtracking-based regex engine.
Such inputs can cause service degradation or outage, as seen in notable incidents at Stack Overflow~\cite{StackOverflow2016REDOSPostMortem} and Cloudflare~\cite{Cloudflare2019REDOSPostMortem}.

\subsubsection{Matching Algorithms}
\label{subsubsec:matching_algorithms}

A \textit{regex engine} compiles a pattern into an intermediate representation and simulates it against input strings.
Two principal algorithms underlie most implementations~\cite{SoKReDoS}.
\textit{Thompson's algorithm}~\cite{thompson1968} performs a breadth-first, lockstep simulation that tracks all active NFA states simultaneously, guaranteeing $O(|Q|^2 \cdot |s|)$ time---linear in input length---but cannot support features requiring memory of previously matched content.
\textit{Spencer's algorithm} performs a depth-first search, exploring one path at a time and backtracking on failure; it accommodates the full E-regex feature set (\eg backreferences, lookarounds) but exhibits worst-case exponential time complexity $O(|Q|^{|s|})$ on pathological inputs.

Most mainstream regex engines---including those in Perl, Python, Java, , and PHP---adopt Spencer-style backtracking because it supports backreferences.
Engines prioritizing worst-case performance (\eg Go, Rust) use Thompson-style matching and thus cannot handle backreferences.

\subsubsection{Ambiguity}
\label{subsubsec:ambiguity}

For an NFA $A$, the \textit{degree of ambiguity}~\cite{weber1991ida} with respect to a string $s$, denoted $\mathrm{AbgS}(A, s)$, is the number of distinct accepting paths for $s$.
The degree of ambiguity for strings of length $n$ is $\mathrm{AbgN}(A, n) = \max_{s \in \varSigma^n} \mathrm{AbgS}(A, s)$, and the overall degree of ambiguity is $\mathrm{Abg}(A) = \max_{n \in \mathbb{N}} \mathrm{AbgN}(A, n)$.

The NFA for \verb|a(b*)c| in \cref{fig:nfa}(1) has \textit{finite} ambiguity: 
for any input \verb|ab|$^n$\verb|c|, exactly one accepting path traverses the \texttt{b}-loop $n$ times.
In contrast, NFAs can exhibit \textit{infinite degree of ambiguity} (IDA), where $\mathrm{Abg}(A) = \infty$.
The NFA for \verb|a*a*| in \cref{fig:nfa}(3) has $n$ accepting paths for input \texttt{a}$^n$---any 
partition between the two loops yields a valid match.
The NFA for \verb$(a|a)*$ in \cref{fig:nfa}(5) has $2^n$ accepting paths, since each \texttt{a} 
can match either branch of the alternation.

\subsubsection{Sink Automaton}
\label{subsubsec:sink_automaton}

To analyze worst-case behavior of a backtracking-based matching algorithm, Weideman \etal~\cite{weideman2016sink} introduced the \textit{sink automaton} $\mathrm{Sink}(A)$, constructed by adding a new accepting state $q_{\mathrm{sink}}$ with $\epsilon$-transitions from every original state and a universal self-loop.
The \textit{sink ambiguity} $\mathrm{SinkAbg}(A) = \mathrm{Abg}(\mathrm{Sink}(A))$ captures all partial matching attempts, not just complete matches.
\cref{fig:nfa}(2), (4), and (6) show the sink automata for the three example regexes.

\subsubsection{Complexity Characterization}
\label{subsubsec:complexity_characterization}

Two theorems connect sink ambiguity to backtracking runtime:

\begin{oldtheorem}[Backtracking Runtime Bound~\cite{weideman2016sink}]
\label{thm:runtime-bound}
For any $\epsilon$-loop-free NFA $A$, its backtracking runtime satisfies
$\mathrm{BtRtN}(A,n) \in \mathcal{O}(\mathrm{SinkAbgN}(A,n))$.
\end{oldtheorem}

\begin{oldtheorem}[Two-Overlap-Loop Characterization~\cite{weber1991ida,alluazen2008}]
\label{thm:two-overlap}
For any $\epsilon$-loop-free NFA $A$, $\mathrm{SinkAbgN}(A,n) \in \Omega(n^2)$ 
if and only if $A$ contains a \textit{two-overlap-loop} structure---two loops sharing 
a common path segment with overlapping accepted symbol sets.
\end{oldtheorem}

For example, the NFA in \cref{fig:nfa}(3) contains two \texttt{a}-loops connected by an $\epsilon$-edge; its sink automaton in \cref{fig:nfa}(4) exhibits $\Theta(n^2)$ ambiguity.
The NFA in \cref{fig:nfa}(5) has two \texttt{a}-loops on the same state, yielding $\Theta(2^n)$ sink ambiguity in \cref{fig:nfa}(6).
When an NFA is trim (\ie all states lie on some accepting path), IDA is equivalent to the presence of a two-overlap-loop structure~\cite{weber1991ida}.

Together, these theorems establish that IDA (equivalently, two-overlap-loop structures) is both necessary and sufficient for super-linear backtracking runtime in K-regexes.
This forms the basis for existing static ReDoS detectors~\cite{weideman2016sink,wustholz2017,hassanImprovingDevelopersUnderstanding2023}.
However, these theorems assume each NFA transition executes in $\mathcal{O}(1)$ time.
This assumption breaks for backreferences, where a single transition may compare substrings of length $\mathcal{O}(n)$, invalidating the runtime bound of \cref{thm:runtime-bound}.
Our work extends this framework to handle such non-constant-cost transitions.

%-------------------------------------------------------------------------------
% \section{Background} % Old
% \label{sec:old_background}
% \input{extra-text/background}
%-------------------------------------------------------------------------------

%-------------------------------------------------------------------------------
% \section{Background} % Old v2
% \label{sec:old_background_v2}
% \input{extra-text/background_v2}
%-------------------------------------------------------------------------------

%-------------------------------------------------------------------------------
\section{Motivation and Problem Statement}
\label{motiv}
%-------------------------------------------------------------------------------

Here we motivate the need for ReDoS analysis beyond K-regexes, present our threat model, and pose research questions. %scope of our investigation.

\subsection{Motivating Example}
\label{sec:motivation:example}

Backreferences are actively used in security-critical regex deployments.
As of November 2025, the Snort intrusion detection system's registered ruleset contains 11,385 unique regexes, of which 278 (2.4\%) use backreferences to describe malicious packet signatures.
In Snort and similar intrusion detection systems, these regexes are evaluated on every inspected packet, making their worst-case performance a security concern: a slow regex evaluation can degrade throughput or cause the system to skip rules entirely.

Existing ReDoS theory provides no guidance on whether backreference-containing regexes are vulnerable.
Existing detectors may discover slow REwB inputs empirically (\eg via fuzzing~\cite{Shen2018ReScueGeneticRegexChecker}), but do not provide structural guarantees or characterize backreference-induced patterns. 
As discussed in \cref{subsec:background_redos_and_regex_complexity}, backtracking runtime for K-regexes is bounded by sink ambiguity (\cref{thm:runtime-bound}), and super-linear runtime occurs only in the presence of two-overlap-loop structures (\cref{thm:two-overlap}).
However, these results assume $\mathcal{O}(1)$-cost transitions.
They become unsound in the presence of backreferences, whose substring comparisons cost $\mathcal{O}(n)$.

To illustrate the gap, consider the regex \verb|(a*)\1b|
%, whose 2PMFA (defined in \cref{sec:2pmfa}) and sink automaton are shown in \cref{fig:sink((a*)1b)}.
This regex contains no two-overlap-loop structure, and its sink ambiguity is 
$\mathcal{O}(n)$---existing detectors would report it as safe.
Yet its backtracking runtime is $\Theta(n^2)$.
Intuitively, on a non-matching input \texttt{a}$^n$, the engine tries each partition of the \texttt{a}s between the capture group \verb|(a*)| and the backreference \verb|\1|: when the group captures $k$ symbols, the backreference performs an $\mathcal{O}(k)$ string comparison before failing, yielding total cost $\sum_{k=1}^{n/2} k = \Theta(n^2)$.
We formalize this analysis in \cref{sec:theory}.

\begin{figure}[!ht]
  \centering
  \includegraphics[trim=0 0 0 50,clip,width=0.8\columnwidth]{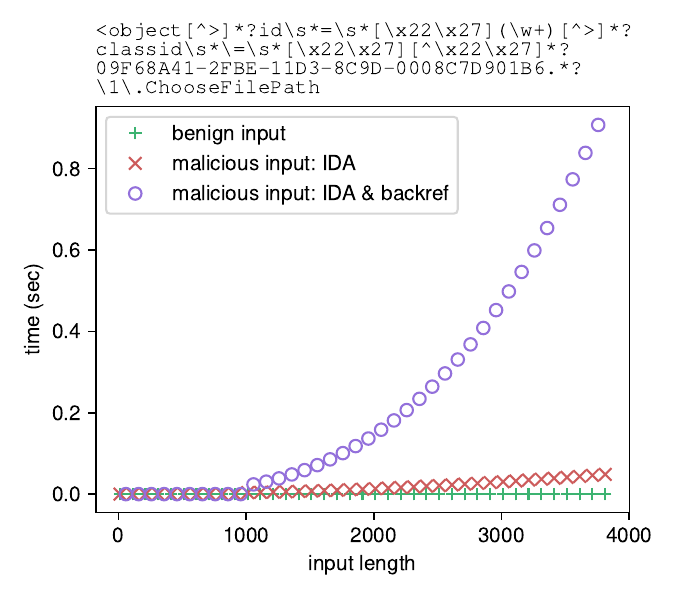}
  \cprotect\caption{Matching time for a regex from the Snort ruleset, evaluated on a benign input and two adversarial inputs exploiting infinite degree of ambiguity (IDA) and a combination of IDA with backreferences. }
  \label{fig:motiv_ida_bref}
\end{figure}

This gap has practical consequences.
\cref{fig:motiv_ida_bref} compares matching time for a Snort regex evaluated on 
benign input, adversarial input exploiting only IDA, and adversarial input 
exploiting IDA with backreference.
The regex contains both: % \verb|\1|:

\begin{lstlisting}[breaklines,breakindent=0ex,basicstyle=\ttfamily\small]
<object[^>]*?id\s*=\s*[\x22\x27](\w+)[^>]*?
classid\s*\=\s*[\x22\x27][^\x22\x27]*?
09F68A41-2FBE-11D3-8C9D-0008C7D901B6.*?
\1\.ChooseFilePath
\end{lstlisting}

\noindent
When only IDA is exploited, runtime grows quadratically.
When both IDA and the backreference are exploited together, runtime grows cubically 
and super-linear behavior is triggered with shorter inputs.
This demonstrates that backreferences represent an additional---and previously 
uncharacterized---attack surface for ReDoS.

\subsection{Threat Model}
\label{sec:motivation:threat}

We adopt a variation of the standard ReDoS threat model from prior 
work~\cite{davis2021,SoKReDoS}, with assumptions tailored to REwB.

\paragraph{Attacker capabilities.}
The attacker controls the input string evaluated by the victim's regex.
This reflects the common use of regexes to process untrusted input in web 
applications~\cite{Chapman2016RegexUsageInPythonApps,barlas_exploiting_2022} and network intrusion detection~\cite{wustholz2017}.
The attacker can also analyze the target regex---for instance, Snort's rulesets are publicly available---to identify exploitable patterns and craft adversarial inputs.

\BC{OK, but prior work also shows that the behavior is mitigated to some extent \cite{SoKReDoS}. We need to address this concern.}
\paragraph{Victim environment.}
The victim uses a backtracking-based regex engine that supports backreferences.
This covers the default engines in Python, Perl, Java, PCRE, PCRE2, and so on~\cite{SoKReDoS}. 
%which collectively dominate server-side and security-tool 
%deployments~\cite{Davis2019LinguaFranca,Staicu2018REDOS}.
Notable exceptions are Rust and Go, which use Thompson-style NFA simulation and do not support backreferences in their default engines.
We note that some engines employ mitigations such as matching limits 
(\eg PCRE's \texttt{pcre\_match\_limit}); as we show in \cref{sec:evaluation}, an attacker can deliberately trigger these limits to cause the engine to abort matching, which itself can be exploited to bypass detection rules.

\paragraph{Attack goal.}
The attacker seeks to cause one of two outcomes:
(1)~\emph{resource exhaustion}, where a crafted input forces the regex engine 
into super-linear evaluation, degrading service availability; or
(2)~\emph{detection bypass}, where the input triggers the engine's matching limit, 
causing it to skip the remainder of the regex and fail to flag malicious content.

\subsection{Research Questions and Scope}
\label{sec:motivation:scope}

Our motivating example reveals that backreferences can induce super-linear backtracking runtime even when existing ReDoS detectors based on NFA-based analyses predict linear behavior.
This leads to three research questions:

\begin{researchquestions}
  \item \textbf{Theory:} \textit{Under what conditions do REwB cause super-linear backtracking runtime, and can we characterize the structural patterns responsible?}
  \item \textbf{Detection:} \textit{Can we develop algorithms that automatically identify vulnerable REwB, and generate corresponding adversarial inputs?}
  \item \textbf{Prevalence and Impact:} \textit{How prevalent are REwB-induced ReDoS vulnerabilities in real-world regex deployments, and what is their practical impact?}
    \begin{researchquestions}
      \item \textit{How many REwB in a real-world deployment are vulnerable to backreference-induced ReDoS?}
      \item \textit{How does matching runtime scale on adversarial inputs, and does combining backreference patterns with IDA worsen the impact?}
      \item \textit{Can REwB vulnerabilities be exploited in a deployed intrusion detection system?}
    \end{researchquestions}
\end{researchquestions}

\noindent
We address RQ1 in \cref{sec:theory}, RQ2 in \cref{sec:vulnerable_rewb_patterns}, and RQ3 in \cref{sec:evaluation}.

\smallskip
\noindent
\textbf{Scope.}
We focus on REwB as defined in \cref{sec:background:regex}, including self-references.
% \BC{I think this is a bit self-conflicting, or needs a better judgment}
% We restrict our analysis to $\epsilon$-loop-free automata, consistent with prior ReDoS theory~\cite{weideman2016sink,wustholz2017}, since $\epsilon$-loops can yield arbitrarily many accepting paths independent of input length.
We do not address other irregular features such as lookaheads or atomic groups, which we leave to future work.
Our evaluation targets the PCRE family and Python's \texttt{re} module, as these are widely used in the security tools (\eg Snort~\cite{Roesch1999Snort}, Suricata~\cite{suricata}) that motivate our study.

%-------------------------------------------------------------------------------
% \section{Motivation} % Old
% \label{sec:old_motivation}
% \input{extra-text/motivation}
%-------------------------------------------------------------------------------

%-------------------------------------------------------------------------------
\section{Two-Phase Memory Finite Automaton} %2PMFA
\label{sec:two_phase_memory_automaton}
%-------------------------------------------------------------------------------

To analyze the backtracking behavior of REwB, we need an automaton model that captures both backreference semantics and self-referencing behavior.
Prior work introduced the Memory Finite Automaton (MFA)~\cite{schmid2016}, which extends NFAs with a memory table that stores captured substrings and replays them on backreference transitions.
However, MFA does not support self-references (\cref{sec:background:regex}).

We propose the \emph{Two-Phase Memory Finite Automaton (2PMFA)}, which extends MFA with a two-phase memory design that cleanly separates the \emph{committed} capture (from the previous iteration of a repeated group) from the \emph{in-progress} capture (being recorded in the current iteration).
This separation enables self-references: when $\backslash i$ is encountered inside group~$i$, the engine matches against the committed phase while the in-progress phase continues recording (\cref{fig:self_bref}).

\subsection{Model Definition}
\label{sec:2pmfa:def}

\begin{definition}{Two-Phase Memory Automaton}{2pmfa}
A 2PMFA is a 6-tuple $A = (Q, \varSigma, I, \varDelta, q_0, F)$ where:
\begin{itemize}[nosep]
    \item $Q$ is a finite set of states,
    \item $\varSigma$ is a finite input alphabet,
    \item $I$ is a finite set of capture group identifiers,
    \item $\varDelta \subseteq Q \times T(A) \times Q$ is the transition relation, with
    \[
        T(A) \;=\;  \{\sigma \in \varSigma\} \;\cup\; \{\varepsilon\} \;\cup\;
        \bigl\{\,\lcg_i,\; \rcg_i,\; \backslash i \;\bigm|\; i \in I\,\bigr\},
    \]
    \item $q_0 \in Q$ is the initial state, and
    \item $F \subseteq Q$ is the set of accepting states.
\end{itemize}
A transition $(q, t, q') \in \varDelta$ moves from state $q$ to $q'$
on label $t$, where labels fall into five categories:
a \emph{symbol} $\sigma \in \varSigma$ consumes one input character;
$\varepsilon$ consumes nothing;
$\lcg_i$ opens capture group $i$ (begins recording);
$\rcg_i$ closes capture group $i$ (commits the recording);
and $\backslash i$ replays the string most recently committed
by group $i$.
\end{definition}

\subsection{Matching Semantics}
\label{sec:2pmfa:algo}

A 2PMFA is matched against an input string $s$ via a backtracking algorithm that maintains a memory function $M$ mapping each capture group to start and end indices into~$s$.
The algorithm explores transitions depth-first, recursively backtracking on failure---mirroring Spencer-style regex engines (\cref{subsubsec:matching_algorithms}).
Symbol, $\varepsilon$, and group-open/close transitions behave as in a standard MFA.
A backreference transition $\backslash i$ compares $s[j \,..\!_< j{+}l]$ against the committed capture $s[M(\lcg_i) \,..\!_< M(\rcg_i)]$, where $l = M(\rcg_i) - M(\lcg_i)$, and advances the input index by~$l$ on success.
The full pseudocode, including the treatment of self-reference semantics is given in \cref{app:matching_algorithm}.

Two properties of this algorithm are critical for the complexity analysis in \cref{sec:theory}:
\begin{enumerate}[nosep,leftmargin=*]
    \item \textbf{Non-constant transition cost.}
    A backreference transition costs $O(l)$ time for a captured substring of length~$l$, which can be as large as $O(n)$. All other transitions cost $O(1)$.
    \item \textbf{Repeated evaluation via backtracking.}
    Backtracking can cause the same transition to be evaluated multiple times across different search branches.
\end{enumerate}
Both properties are absent in standard NFA simulation and are the root cause of backreference-induced ReDoS.

\subsection{Path Notation}
\label{sec:2pmfa:path}

We establish path notation used throughout the complexity analysis (\cref{sec:theory}) and vulnerability pattern classification (\cref{sec:vulnerable_rewb_patterns}).

\begin{definition}{2PMFA Path}{path}
Given a 2PMFA $A = (Q, \varSigma, I, \varDelta, q_0, F)$,
a \emph{path}~$\pi$ from $q'_0$ to $q'_m$ is a sequence
\[
    q'_0 \xrightarrow[s'_0]{t'_0} q'_1
         \xrightarrow[s'_1]{t'_1} \cdots\;
    q'_{m-1} \xrightarrow[s'_{m-1}]{t'_{m-1}} q'_m
\]
where each step satisfies
$(q'_k, t'_k, q'_{k+1}) \in \varDelta$ and $t'_k$ matches $s'_k$.
\end{definition}

\noindent We use the following conventions:
\begin{itemize}[nosep]
\item \textbf{Accepting path:}
    $\pi$ is accepting if $q'_0 = q_0$ and $q'_m \in F$.
\item \textbf{String of a path:}
    $\mathcal{S}(\pi) = s'_0 s'_1 \cdots s'_{m-1}$.
    We write $q'_0 \overset{\pi}{\underset{s}{\leadsto}} q'_m$
    when $\mathcal{S}(\pi) = s$.
\item \textbf{Loop path:}
    A path from $q$ back to $q$ is denoted $\pi^*$
    (emphasizing its role as a repeatable loop).
\item \textbf{Backreference cost:}
    When a step has label $t' = \backslash i$,
    it matches $s'_k$ of length up to $O(n)$
    and costs $O(|s'_k|)$ time (contrasting with $\sigma$ and $\varepsilon$ transitions,
    which cost $O(1)$).
\item \textbf{Path overlap:}
    We say $\pi_1, \cdots, \pi_m$ overlap when their strings are formed by the same repeated substring. Formally, $\mathrm{Ovlp}(\pi_1, \cdots, \pi_m)$ iff $\exists s_\mathrm{ovlp} \in \varSigma^*, u_1, \cdots, u_m \in \mathbb{N}$, s.t.\ for $k \in \{1,\cdots, m\}$, $\mathcal{S}(\pi_k) = s_\mathrm{ovlp}^{u_k}$.
\end{itemize}

%\noindent
%The distinction between $O(1)$-cost and $O(n)$-cost transitions is central to our analysis: it is precisely why backreferences can cause super-linear runtime even when the sink ambiguity of the underlying automaton is only linear.

%-------------------------------------------------------------------------------
% \section{Two-Phase Memory Finite Automaton} % Old
% \label{sec:old_two_phase_memory_automaton}
% \input{extra-text/two_phase_memory_automaton}
%-------------------------------------------------------------------------------

%-------------------------------------------------------------------------------
% \section{Regular Expression with Backreferences} % Old
% \label{sec:old_rewb}
% \input{extra-text/rewb}
%-------------------------------------------------------------------------------

%-------------------------------------------------------------------------------
\section{Theoretical Analysis of REwB}
\label{sec:theory}
%-------------------------------------------------------------------------------

This section lays out the theoretical foundations for ReDoS vulnerabilities caused by REwB.
We begin with a concrete example showing that existing runtime bounds fail for REwB (\cref{sec:theory:ex}).
We then identify two independent conditions that enable non-$\mathcal{O}(1)$ per-backreference matching cost (\cref{sec:theory:conditions}).
From these conditions, we derive sufficient conditions under which the existing runtime bound still holds (\cref{sec:theory:safe}), and necessary conditions under which it is violated (\cref{sec:theory:unsafe}). 
% The latter directly motivates the vulnerable patterns classified in \cref{sec:patterns}.

\subsection{Why Existing Bounds Fail}
\label{sec:theory:ex}

Recall from \cref{subsec:background_redos_and_regex_complexity} that for K-regexes, \cref{thm:runtime-bound} bounds backtracking runtime by sink ambiguity, and \cref{thm:two-overlap} shows that super-linear sink ambiguity requires two overlapping loops (the IDA condition).
The following example demonstrates that \emph{neither} conditions is necessary for super-linear runtime when backreferences are present.

\begin{example}{}{a-star-backref}
    Let $A$ be the 2PMFA for
    $\lcg_1 \texttt{a}^* \rcg_1 \backslash\!1\;\texttt{b}$.
    Then, $\mathrm{AbgN}(A, n) \in \mathcal{O}(1)$,\;
    $\mathrm{SinkAbgN}(A, n) \in \mathcal{O}(n)$,\;
    yet $\mathrm{BtRtN}(A, n) \in \Theta(n^2)$.
\end{example}

\begin{figure}[t]
    \centering
    \includegraphics[width=\linewidth]{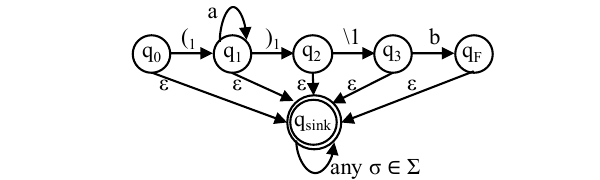}
  
\iffalse
    {\small
    \begin{tikzpicture}
        \node[shape=circle, draw] (q0) at (0, 0) {$q_0$};
        \node[shape=circle, draw] (q1) at (4em, 0) {$q_1$};
        \node[shape=circle, draw] (q2) at (8em, 0) {$q_2$};
        \node[shape=circle, draw] (q3) at (12em, 0) {$q_3$};
        \node[shape=circle, draw] (qF) at (16em, 0) {$q_F$};
        \node[shape=circle, draw, double] (qsink) at (8em, -4em) {$q_\mathrm{sink}$};
        \draw[->] (q0) -- (q1) node[pos=0.5, above] {$\lcg_1$};
        \draw[->] (q1.north west) to [out=120, in=70, looseness=3]
            node [pos=0.5, above] {\texttt{a}}  (q1.north east);
        \draw[->] (q1) -- (q2) node [pos=0.5, above] {$\rcg_1$};
        \draw[->] (q2) -- (q3) node [pos=0.5, above] {$\backslash 1$};
        \draw[->] (q3) -- (qF) node [pos=0.5, above] {\texttt{b}};
        \draw[->] (q0) -- (qsink) node [pos=0.5, below] {$\epsilon$};
        \draw[->] (q1) -- (qsink) node [pos=0.5, above] {$\epsilon$};
        \draw[->] (q2) -- (qsink) node [pos=0.5, left] {$\epsilon$};
        \draw[->] (q3) -- (qsink) node [pos=0.5, above] {$\epsilon$};
        \draw[->] (qF) -- (qsink) node [pos=0.5, below] {$\epsilon$};
        \draw[->] (qsink.south west) to [out=-120, in=-70, looseness=3]
            node [pos=0.5, below] {any $\sigma \in \varSigma$} (qsink.south east);
    \end{tikzpicture}
    }
\fi
    \caption{Sink automaton $\mathrm{Sink}(A)$ for the regex
    $\lcg_1 \texttt{a}^* \rcg_1 \backslash\!1\;\texttt{b}$.
    The original automaton $A$ has a single \texttt{a}-loop
    at~$q_1$; no two overlapping loops exist.}
    \label{fig:sink-example}
\end{figure}

\begin{proof}
\cref{fig:sink-example} shows $\mathrm{Sink}(A)$. The original automaton $A$ contains a single loop (the \texttt{a}-loop at~$q_1$); no two overlapping loops exist.

\paragraph{Ambiguity is $\mathcal{O}(1)$.}
$A$ accepts only strings of the form $\texttt{a}^{n}\texttt{b}$. %: the loop captures $\texttt{a}^m$, the backreference replays it, and \texttt{b} terminates the match.
For such an input, there exists exactly one accepting path: 
\[
q_0 \xrightarrow{\lcg_1} \pi^*_{\mathrm{a}}
\xrightarrow{\rcg_1} q_2
\xrightarrow[\texttt{a}^{n/2}]{\textbackslash 1} q_3
\xrightarrow{\texttt{b}} q_F,
\quad \text{where } \mathcal{S}(\pi^*_{\mathrm{a}}) = \texttt{a}^{n/2}.
\]
Thus, $\mathrm{AbgN}(A,n) \in \mathcal{O}(1)$.

\paragraph{Sink ambiguity is $\mathcal{O}(n)$.}
On input $\texttt{a}^n\texttt{b}$, the sink automaton admits
the following families of accepting paths:
\begin{enumerate}[nosep,label=(\roman*)]
    \item Entering $q_\mathrm{sink}$ directly from~$q_0$
        via~$\varepsilon$ (1~path):
        \[
            q_0 \xrightarrow{\varepsilon}
            \pi_\mathrm{sink_1}^*,
            \quad\text{where }
            \mathcal{S}(\pi_\mathrm{sink_1}^*)
            = \texttt{a}^n\texttt{b}.
        \]

    \item Looping $k$ times at~$q_1$ and then
        entering~$q_\mathrm{sink}$ from $q_1$ or~$q_2$
        ($2(n{+}1)$~paths):
        \[
            q_0 \xrightarrow{\lcg_1}
            \pi_{\texttt{a}_2}^*
            \xrightarrow{\varepsilon}
            \pi_\mathrm{sink_2}^*
            \quad\text{and}\quad
            q_0 \xrightarrow{\lcg_1}
            \pi_{\texttt{a}_2}^*
            \xrightarrow{\rcg_1} q_2
            \xrightarrow{\varepsilon}
            \pi_\mathrm{sink_2}^*,
        \]
        where
        $\mathcal{S}(\pi_{\texttt{a}_2}^*) = \texttt{a}^k$,\;
        $\mathcal{S}(\pi_\mathrm{sink_2}^*)
        = \texttt{a}^{n-k}\texttt{b}$,\;
        $k \in \mathbb{N}_{0..n}$.

    \item Capturing $\texttt{a}^k$, matching the
        backreference, and entering~$q_\mathrm{sink}$
        from~$q_3$ ($\lfloor n/2 \rfloor{+}1$~paths):
        \[
            q_0 \xrightarrow{\lcg_1}
            \pi_{\texttt{a}_3}^*
            \xrightarrow{\rcg_1} q_2
            \xrightarrow[s_{\backslash 1}]{\backslash 1}
            q_3
            \xrightarrow{\varepsilon}
            \pi_\mathrm{sink_3}^*,
        \]
        where
        $\mathcal{S}(\pi_{\texttt{a}_3}^*)
        = s_{\backslash 1} = \texttt{a}^k$,\;
        $\mathcal{S}(\pi_\mathrm{sink_3}^*)
        = \texttt{a}^{n-2k}\texttt{b}$,\;
        $k \in \mathbb{N}_{0..n/2}$.

    \item The unique full accepting path
        through~$q_F$ (1~path):
        \[
            q_0 \xrightarrow{\lcg_1}
            \pi_{\texttt{a}_4}^*
            \xrightarrow{\rcg_1} q_2
            \xrightarrow[\texttt{a}^{n/2}]{\backslash 1}
            q_3
            \xrightarrow{\texttt{b}} q_F
            \xrightarrow{\varepsilon}
            q_\mathrm{sink}.
        \]
\end{enumerate}
In total there are $(3n/2)+2$ accepting paths, so
$\mathrm{SinkAbgN}(A,n) \in \mathcal{O}(n)$.

\paragraph{Runtime is $\Theta(n^2)$.}
Consider the input $\texttt{a}^n$ (no trailing~\texttt{b}; the match will ultimately fail).
With a greedy loop, the engine first tries capturing all $n$
symbols, then backtracks one symbol at a time.
\cref{tbl:runtime-example} summarizes the cost of each attempt. 
When the loop captures $\texttt{a}^k$ ($k > n/2$), the backreference fails in $\mathcal{O}(1)$ time because fewer than~$k$ symbols remain.
When $k \le n/2$, the backreference performs a full $\mathcal{O}(k)$ string comparison before the suffix \texttt{b} mismatches.
The total cost is:
\[
    n + \Bigl(\frac{n}{2} - 1\Bigr)
    + \sum_{k=1}^{n/2} k
    + 1
    \;=\; \frac{n^2}{8} + \frac{7n}{4}
    \;\in\; \Theta(n^2).
\]

\noindent
The two root causes of this quadratic blowup are:
\begin{enumerate}[leftmargin=*,nosep]
\item \textbf{Non-$\mathcal{O}(1)$ per-backreference cost.}
    The backreference $\backslash\!1$ matches a captured substring of length up to $n/2$, so a single evaluation costs $\mathcal{O}(n)$.
\item \textbf{Non-$\mathcal{O}(1)$ evaluation count.}
    Backtracking causes $\backslash\!1$ to be evaluated $\Theta(n)$ times (once per loop iteration that is retried).
\end{enumerate}
Together these yield $\mathcal{O}(n) \times \mathcal{O}(n) = \mathcal{O}(n^2)$ runtime, violating \cref{thm:runtime-bound} despite $\mathcal{O}(n)$ sink ambiguity.
We formalize each condition in \cref{sec:theory:conditions}.

\end{proof}

\begin{table}[h]
    \centering
    \caption{Runtime analysis for
    $\lcg_1 \texttt{a}^* \rcg_1 \backslash\!1\;\texttt{b}$
    on input~$\texttt{a}^n$.  A dash indicates that the
    backreference fails in $\mathcal{O}(1)$ time
    (remaining input shorter than capture).}
    \label{tbl:runtime-example}
    \small
    \begin{tabular}{c c c c}
        \toprule
        Attempt & Loop captures & $\backslash\!1$ matches & Cost \\
        \midrule
        $0$         & $\texttt{a}^n$       & ---                  & $n$     \\
        $1$         & $\texttt{a}^{n-1}$   & ---                  & $1$     \\
        $\vdots$    &                      &                      &         \\
        $n/2{-}1$   & $\texttt{a}^{n/2+1}$ & ---                  & $1$     \\
        % \midrule
        $n/2$       & $\texttt{a}^{n/2}$   & $\texttt{a}^{n/2}$   & $n/2$   \\
        $n/2{+}1$   & $\texttt{a}^{n/2-1}$ & $\texttt{a}^{n/2-1}$ & $n/2{-}1$ \\
        $\vdots$    &                      &                      &         \\
        $n{-}1$     & $\texttt{a}^{1}$     & $\texttt{a}^{1}$     & $1$     \\
        $n$         & $\varepsilon$        & $\varepsilon$        & $1$     \\
        \bottomrule
    \end{tabular}
\end{table}

\subsection{Conditions for Super-Linear REwB} % Behavior}
\label{sec:theory:conditions}

\cref{ex:a-star-backref} revealed two independent factors that cause \cref{thm:runtime-bound,thm:two-overlap} to fail:
    a single backreference evaluation may cost non-$\mathcal{O}(1)$ time, 
    and a backreference may be evaluated a non-$\mathcal{O}(1)$ number of times.
We now formalize each condition as a lemma.

\paragraph{Per-evaluation cost.}
In a standard NFA, every transition consumes exactly one symbol or~$\varepsilon$, so each step costs~$\mathcal{O}(1)$.
A backreference transition~$\backslash\!i$, however, performs a string comparison against the captured content of group~$i$, which may have length up to~$\mathcal{O}(n)$.
\cref{lem:non-O1-cost} identifies the structural condition that permits this.

\begin{lemma}{Non-$\mathcal{O}(1)$ Per-Backreference Cost}{non-O1-cost}
    For an $\varepsilon$-loop-free 2PMFA~$A$, if a backreference transition $\backslash\!i$ can match a string of non-$\mathcal{O}(1)$ length, then capture group~$i$ must contain either: a loop, or a backreference that itself matches a string of non-$\mathcal{O}(1)$ length.
\end{lemma}

\iffalse
\begin{proof}
    Among the five transition types in a 2PMFA
    % (symbol, $\varepsilon$, group-open, group-close, backreference),
    , only a backreference can match strings of unbounded length; all others match at most one symbol.
    %
    If capture group~$i$ contains no backreference (case~(a)), then every transition inside it matches $\mathcal{O}(1)$ symbols.
    Suppose, for contradiction, that no transition appears more than once on any path through the group.
    Since the number of transitions is fixed, the captured string has length~$\mathcal{O}(1)$---a contradiction.
    Hence, some transition must repeat, which requires a subpath from some state~$q$ back to~$q$: a loop.
    
    If capture group~$i$ does contain a backreference (case~(b)), that inner backreference must itself reference a capture group satisfying (a) or~(b), producing a recursive or chain dependency that ultimately terminates at a loop.
    %
    \cref{fig:lemma}(a) illustrates both cases.
\end{proof}
\fi

\begin{proof}
    Among the five transition types in a 2PMFA
    % (symbol, $\varepsilon$, group-open, group-close, backreference),
    , only a backreference can match strings of unbounded length (\ie non-$\mathcal{O}(1)$); all others match at most one symbol.
    Given this, there are two cases in which a backreference
    $\backslash\!i$ matches a string of non-$\mathcal{O}(1)$
    length.

    \emph{\textnormal{\textbf{Case~1:}} Capture group~$i$ contains no
    backreference transitions (\ie every transition inside
    the group matches $\mathcal{O}(1)$-length strings)}.
    We show by contradiction that some transition must appear
    more than once on a path through the group.
    If each transition appeared at most once, then because
    the total number of transitions is $\mathcal{O}(1)$,
    any captured string would have length~$\mathcal{O}(1)$---a
    contradiction.
    If a transition appears more than once along a path,
    the path must take the form:
    \[
        \pi_\mathrm{left}\;
        q \xrightarrow[s]{t} q'\;
        \pi_\mathrm{pump}\;
        q \xrightarrow[s]{t} q'\;
        \pi_\mathrm{right}
    \]
    which exhibits a subpath $\pi_\mathrm{pump}$ from~$q$
    back to~$q$ (\ie a loop).
    In \cref{fig:lemma}(a), the backreference
    $\backslash\!1$ incurs non-$\mathcal{O}(1)$ cost when
    matching capture group~1 that contains such a loop.

    \emph{\textnormal{\textbf{Case~2:}} Capture group~$i$ contains a
    backreference that matches strings of
    non-$\mathcal{O}(1)$ length.}
    For this to occur, the inner backreference must reference
    another capture group that itself contains a loop
    (reducing to Case~1 or a further backreference capable
    of matching non-$\mathcal{O}(1)$ strings, or the capture
    groups form a cyclic chain of references that ultimately
    terminates at a loop.
    In \cref{fig:lemma}(a), the backreference
    $\backslash\!2$ incurs non-$\mathcal{O}(1)$ cost when
    matching a capture group that contains~$\backslash\!1$.
\end{proof}

\paragraph{Evaluation count.}
Even when each backreference evaluation is cheap, the \emph{number} of evaluations may be super-linear due to backtracking.
\cref{lem:non-O1-evals} formalizes this.

\begin{lemma}{Non-$\mathcal{O}(1)$ Backreference Evaluations}{non-O1-evals}
    For an $\varepsilon$-loop-free 2PMFA $A = (Q, \varSigma, I, \varDelta, q_0, F)$, if a transition $\delta = ((q,t) \mapsto Q') \in \varDelta$ is evaluated a non-$\mathcal{O}(1)$ number of times, then there exists a path in~$A$ in which $\delta$ appears after or inside a loop.
\end{lemma}

\begin{proof}
    There are two cases where $\delta$ is evaluated a non-$\mathcal{O}(1)$ number of times.

    \emph{\textnormal{\textbf{Case~1:}} $\delta$ is evaluated across non-$\mathcal{O}(1)$ backtracking branches.}
    This implies that there exist a non-$\mathcal{O}(1)$ number of distinct paths from~$q_0$ to~$q$.
    We prove by contradiction that within such paths, there must exist a transition $\delta_1 = ((q_1, t_1) \mapsto Q'_1)$ that appears more than once before~$\delta$.
    Assume instead that each transition appears at most once along any such path.
    Then the maximum number of paths from~$q_0$ to~$q$ would be $\sum_{k=0}^{|\varDelta|-1} \mathrm{P}^k_{|\varDelta|-1}$,
    which is $\mathcal{O}(1)$ with respect to the input length---contradicting the assumption.
    Therefore, $\delta_1$ must occur multiple times on some path, implying the existence of a subpath from~$q_1$
    to~$q_1$, \ie a loop before~$\delta$.
    In \cref{fig:lemma}(b), the backreference $\backslash\!1$ may be evaluated a non-$\mathcal{O}(1)$ number of times after such a loop.

    \emph{\textnormal{\textbf{Case~2:}} $\delta$ is evaluated non-$\mathcal{O}(1)$ times within a single backtracking path.}
    This means that along a single path starting from~$q_0$, the transition~$\delta$ appears non-$\mathcal{O}(1)$ times.
    Consequently, the path must contain a subpath from~$q$ back to~$q$ (\ie a loop) in which $\delta$ is contained.
    In \cref{fig:lemma}(c), the backreference $\backslash\!1$ may be evaluated a non-$\mathcal{O}(1)$ number of times within a loop.
\end{proof}

\begin{figure}[h]
  \centering
  \includegraphics[width=\linewidth]{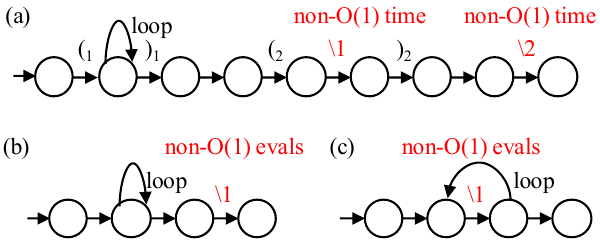}
  \caption{
      Structural conditions for super-linear backreference behavior.
      (a)~A backreference incurs non-$\mathcal{O}(1)$ cost when its capture group contains a loop (left) or another non-$\mathcal{O}(1)$ backreference (right). 
      (b)--(c)~A backreference is evaluated non-$\mathcal{O}(1)$ times when it appears after a loop~(b) or inside a loop~(c).
  }
  \label{fig:lemma}
\end{figure}

\subsection{Conditions for Bounded Runtime}
%\subsection{Sufficient Conditions for Bounded Runtime}
\label{sec:theory:safe}

The \cref{lem:non-O1-cost,lem:non-O1-evals} identify \emph{what can go wrong}.
We now show that if neither condition is fully triggered on each transition, the classical runtime bound (\cref{thm:runtime-bound}) continues to hold for REwB.

\begin{theorem}{Safe Backreferences}{safe-backref}
    For an $\varepsilon$-loop-free 2PMFA~$A$ running on strings of length~$n$: if every backreference in~$A$ either
      (i)~captures a string of length $\mathcal{O}(1)$, or
      (ii)~is evaluated a total of $\mathcal{O}(1)$ times,
    then
    $\mathrm{BtRtN}(A,n) \in \mathcal{O}(\mathrm{SinkAbgN}(A,n))$.
\end{theorem}

\begin{proof}[\textit{Proof sketch}]
    (Full proof in \cref{sec:proofthm1}.)
    We define algorithms for computing sink ambiguity ($\mathrm{SinkAbgS}$, \cref{algo:SinkAbgS}) and backtracking runtime ($\mathrm{BtRtS}$, \cref{algo:BtRtS}).
    Our goal is to construct a scaled-up (upper-approximate) version $\mathrm{BtRtS}\!\!\uparrow(A,s)$ satisfying, for some constant~$\xi$,
    \[
        \mathrm{BtRtS}(A,s)
        \;\le\;
        \mathrm{BtRtS}\!\!\uparrow(A,s)
        \;\le\;
        \xi \cdot \mathrm{SinkAbgS}(A,s).
    \]

    \emph{Bounded-length backreferences.}
    Let $\mathrm{MaxFBrL}(A)$ be the maximum length matched by any $\mathcal{O}(1)$-length backreference, and $\mathrm{MaxOut}(A)$ the maximum out-degree of any state.
    Both are constants with respect to~$|s|$.
    We show that each transition contributes at most $\mathrm{MaxOut}(A) \cdot \mathrm{MaxFBrL}(A)$ to runtime.

    \emph{Unbounded-length backreferences.}
    By condition~(ii), such backreferences are evaluated $\mathcal{O}(1)$ times in total.
    Let $\mathrm{IBrRCt}(A)$ denote the maximum total number of these evaluations; then their aggregate cost is at most $\mathrm{IBrRCt}(A) \cdot |s|$.

    \emph{Combining both.}
    The scaled runtime is (\cref{algo:BtRtU}):
    \begin{align*}
        \mathrm{BtRtS}\!\!\uparrow(A,s)
        \;=\;&
        \mathrm{MaxOut}(A) \cdot \mathrm{MaxFBrL}(A)
            \cdot \mathrm{SinkAbgS}(A,s) \\
        &+\; \mathrm{IBrRCt}(A) \cdot |s|.
    \end{align*}
    Since
    $\mathrm{BtRtS}(A,s) \le
     \mathrm{BtRtS}\!\!\uparrow(A,s)$,
    it suffices to compare
    $\mathrm{BtRtS}\!\!\uparrow(A,s)$ with
    $\mathrm{SinkAbgS}(A,s)$ across three possible growth
    cases of $\mathrm{SinkAbgS}(A,s)$.
    In each valid case, we show the existence of a constant~$\xi$ such that
    \[
        \mathrm{BtRtS}(A,s)
        \;\le\;
        \xi \cdot \mathrm{SinkAbgS}(A,s).
    \]
    Finally, by taking the maximum over all strings of length~$n$, we obtain
    \[
        \mathrm{BtRtN}(A,n) \in
        \mathcal{O}(\mathrm{SinkAbgN}(A,n)). \qedhere
    \]
\end{proof}

\subsection{Necessary Conditions for Vulnerability}
\label{sec:theory:unsafe}

Taking the contrapositive of our own \cref{thm:safe-backref}, we obtain the structural conditions that \emph{must} hold whenever backreferences cause the runtime to exceed the sink-ambiguity bound.

\begin{theorem}{Necessities for REwB Vulnerability}{necessary}
    For an $\varepsilon$-loop-free 2PMFA~$A$ running on
    strings of length~$n$: if
    $\mathrm{BtRtN}(A,n) \notin
     \mathcal{O}(\mathrm{SinkAbgN}(A,n))$,
    then there exists a backreference transition
    $\backslash\!i$ satisfying \textbf{both}:
    \begin{enumerate}[nosep,label=\textbf{C\arabic*}.]
    \item \emph{Non-$\mathcal{O}(1)$ cost.}\;
        Capture group~$i$ contains a loop or a backreference
        that matches a string of non-$\mathcal{O}(1)$ length.
        (\cref{lem:non-O1-cost}; \cref{fig:lemma}(a))
    \item \emph{Non-$\mathcal{O}(1)$ evaluations.}\;
        $\backslash\!i$ appears after or inside a loop.
        (\cref{lem:non-O1-evals}; \cref{fig:lemma}(b--c))
    \end{enumerate}
\end{theorem}

\begin{proof}
    The contrapositive of \cref{thm:safe-backref} is:
    if $\mathrm{BtRtN}(A,n) \notin \mathcal{O}(\mathrm{SinkAbgN}(A,n))$, then some backreference simultaneously matches non-$\mathcal{O}(1)$-length strings \emph{and} is evaluated non-$\mathcal{O}(1)$ times.
    Applying \cref{lem:non-O1-cost} to the first conjunct yields~\textbf{C1}.
    Applying \cref{lem:non-O1-evals} to the second yields~\textbf{C2}.
\end{proof}

\cref{thm:necessary} reduces vulnerability detection to a structural search problem over 2PMFA paths.
Any REwB whose backtracking runtime exceeds its sink ambiguity must contain a backreference satisfying both~\textbf{C1} and~\textbf{C2}.
%Our \cref{thm:necessary} is the outcome that makes vulnerability detection possible: any REwB whose backtracking runtime exceeds its sink ambiguity must contain a backreference satisfying both~\textbf{C1} and~\textbf{C2}.
In particular, when the sink ambiguity is $\mathcal{O}(n)$ (\ie no two-overlap loops exist and existing detectors report \emph{no vulnerability}), \textbf{C1}~and~\textbf{C2} together can still produce $\Omega(n^2)$ runtime---as demonstrated by \cref{ex:a-star-backref}.
We exploit this characterization in \Cref{sec:vulnerable_rewb_patterns} to derive three concrete vulnerability patterns and prove that each is sufficient to induce super-linear runtime.

\paragraph{Role of two-overlap loops.}
\BC{I say let's cut this paragraph}
When \cref{thm:safe-backref} \emph{does} hold (\ie backreferences are safe), super-linear runtime still requires super-linear sink ambiguity.
By \cref{thm:two-overlap}, this is equivalent to the presence of two overlapping loops---the classical IDA condition.
In other words, safe backreferences do not introduce new vulnerability patterns beyond those already detectable by existing K-regex tools.

\paragraph{Scope and completeness.}
The conditions in \cref{thm:necessary} are necessary but not sufficient: not every backreference satisfying \textbf{C1}~and~\textbf{C2} induces super-linear runtime.
The three patterns we derive in \Cref{sec:vulnerable_rewb_patterns} are each proven sufficient, but may not form a complete characterization.
Additionally, our analysis assumes $\mathcal{O}(n)$ sink ambiguity.
Extending the structural equivalence between sink ambiguity and overlap loops (\cref{thm:two-overlap}) from NFAs to 2PMFAs remains open; we conjecture that backreferences do not introduce additional sink ambiguity, but leave formal proof to future work. % (\cref{sec:discussion}).

\begin{rqanswer}[RQ1 (Theory)]
REwB cause super-linear backtracking runtime when a backreference satisfies two conditions simultaneously: \textbf{(C1)} its capture group contains a loop, enabling non-$\mathcal{O}(1)$ match length per evaluation; and \textbf{(C2)} it appears after or inside a loop, enabling non-$\mathcal{O}(1)$ total evaluations during backtracking.
When both conditions hold, the product of per-evaluation cost and evaluation count yields super-linear runtime, even when sink ambiguity remains $\mathcal{O}(n)$.
\end{rqanswer}

%-------------------------------------------------------------------------------
% \section{Time Complexity of Regex Matching with Backreferences} % Old
% \label{sec:old_theory}
% \input{extra-text/time_complexity_of_regex_matching_with_rewb}
%-------------------------------------------------------------------------------

%-------------------------------------------------------------------------------
\section{Vulnerable REwB Patterns}
\label{sec:vulnerable_rewb_patterns}
%-------------------------------------------------------------------------------

In \cref{sec:theory}, we set the necessary conditions under which backreferences cause the backtracking runtime to exceed the sink-ambiguity bound (\Cref{thm:necessary}).
In this section, we derive three concrete vulnerability patterns from those conditions and prove that each is \emph{sufficient} to induce super-linear runtime, even when the sink ambiguity is $\mathcal{O}(n)$---i.e., when no double-overlap-loop (IDA) pattern exists. % (violating literature's \cref{thm:two-overlap}).
We begin by classifying the patterns (\S\ref{sec:pattern-class}), then prove their sufficiency (\S\ref{sec:sufficiency}), and finally show that the three patterns exhaustively cover the cases implied by \Cref{thm:necessary} (\S\ref{sec:necessity-logic}).

\subsection{Pattern Classification} \label{sec:pattern-class}

\Cref{thm:necessary} requires two conditions to hold simultaneously for a backreference to cause unbounded runtime:

\begin{enumerate}[label=\textbf{C\arabic*},leftmargin=*]
  \item \textbf{Non-$\mathcal{O}(1)$ per-evaluation cost} (\Cref{lem:non-O1-cost}): the referenced capture group must contain a loop $\pi_\mathrm{pump}$ (or another non-$\mathcal{O}(1)$-length backreference), enabling the captured string to grow with input length. \Cref{fig:patterns}(1) shows the generalized sub-pattern: a capture group delimited by $\lcg_i$ and $\rcg_i$, with a left path~$\pi_\mathrm{left}$, a loop~$\pi_\mathrm{pump}$, and a right path~$\pi_\mathrm{right}$.

  \item \textbf{Non-$\mathcal{O}(1)$ evaluation count} (\Cref{lem:non-O1-evals}): the backreference must appear after or inside a loop~$\pi_\mathrm{loop}$, so that it is evaluated a non-$\mathcal{O}(1)$ number of times during backtracking. \Cref{fig:patterns}(2) shows this sub-pattern: a loop $\pi_\mathrm{loop}$ connected to the backreference via a bridge path~$\pi_\mathrm{bridge}$.
\end{enumerate}

\noindent
Additionally, because we restrict attention to the $\mathcal{O}(n)$ sink-ambiguity regime (no IDA), the loop~$\pi_\mathrm{loop}$, bridge~$\pi_\mathrm{bridge}$, and the backreference must all accept a common overlap string~$s_\mathrm{ovlp}$.
Without this overlap, the loop cannot produce a non-$\mathcal{O}(1)$ number of \emph{distinct} path decompositions prior to the backreference while staying below IDA.
Intuitively, the overlap allows the input to be partitioned in multiple ways between $\pi_\mathrm{loop}$ and the backreference---\eg for an input $s_\mathrm{ovlp}^n$, the loop may consume between $0$ and $n$ copies, while the backreference matches the corresponding captured substring.

\begin{figure}[t]
  \centering
  \includegraphics[width=\linewidth]{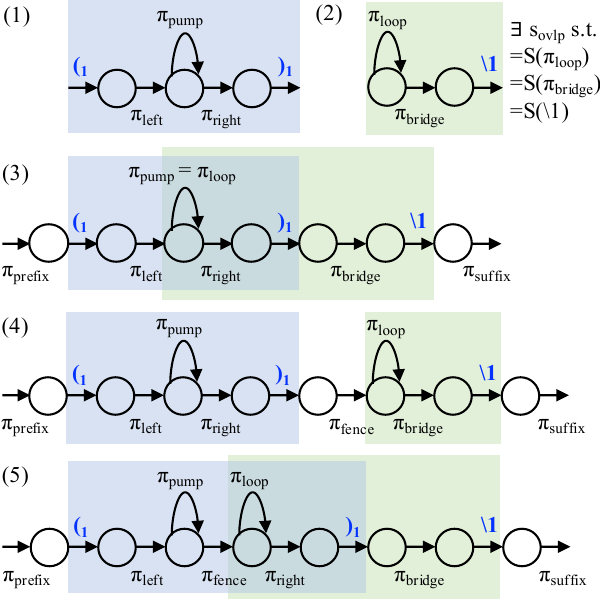}
  \cprotect\caption{Sub-patterns for \textbf{C1} (1) and \textbf{C2} (2), and the three vulnerability patterns (3--5) derived by composing them.}
  \label{fig:patterns}
\end{figure}

The two conditions can be composed in exactly three structurally distinct ways, depending on 
    (i)~whether $\pi_\mathrm{pump}$ and $\pi_\mathrm{loop}$ are the same loop or distinct, 
    and (ii)~whether $\pi_\mathrm{loop}$ lies inside or outside the capture group.
We define each pattern:

% ---- Pattern 1 ----
\begin{pattern}{Backref-to-Overlap-Loop}{1}
A 2PMFA contains Pattern~1 if it has a path of the form
\[
  \pi_\mathrm{prefix}
  \;\xrightarrow{\lcg_i}\;
  \pi_\mathrm{left}\;\pi_\mathrm{pump}^*\;\pi_\mathrm{right}
  \;\xrightarrow{\rcg_i}\;
  \pi_\mathrm{bridge}
  \;\xrightarrow{\backslash i}\;
  \pi_\mathrm{suffix}
\]
where $\pi_\mathrm{pump}$ serves as \emph{both} the pump loop (\textbf{C1}) and the evaluation loop (\textbf{C2}), and the overlap condition is:
%$s_\mathrm{ovlp} = S(\pi_\mathrm{left}) = S(\pi_\mathrm{pump}) = S(\pi_\mathrm{right} \pi_\mathrm{bridge})$.
$\mathrm{Ovlp}(\pi_\mathrm{left},\, \pi_\mathrm{pump},\, \pi_\mathrm{right} \pi_\mathrm{bridge})$.
\end{pattern}

\noindent
\Cref{fig:patterns}(3) illustrates Pattern~1.
Here, a single loop plays both roles: it inflates the captured string (satisfying \textbf{C1}) and, because the backreference appears after the same loop, generates multiple backtracking paths (satisfying \textbf{C2}).
Because only one loop is involved, no double-overlap-loop structure exists, and existing ReDoS detectors cannot flag this pattern.
% Example~1 in \cref{sec:example-redos-backref} is an instance of Pattern~1: the regex $\lcg_1\texttt{a}^*\rcg_1\backslash\texttt{1b}$ has a single \texttt{a}-loop inside the capture group, with the backreference immediately after it.

% ---- Pattern 2 ----
\begin{pattern}{Loop-Before-Backref-to-Loop}{2}
A 2PMFA contains Pattern~2 if it has a path of the form
\[
  \pi_\mathrm{prefix}
  \;\xrightarrow{\lcg_i}\;
  \pi_\mathrm{left}\;\pi_\mathrm{pump}^*\;\pi_\mathrm{right}
  \;\xrightarrow{\rcg_i}\;
  \pi_\mathrm{fence}\;
  \pi_\mathrm{loop}^*\;\pi_\mathrm{bridge}
  \;\xrightarrow{\backslash i}\;
  \pi_\mathrm{suffix}
\]
where $\pi_\mathrm{pump}$ and $\pi_\mathrm{loop}$ are \emph{distinct} loops, $\pi_\mathrm{pump}$ is inside the capture group (\textbf{C1}), $\pi_\mathrm{loop}$ is \emph{outside} the capture group (\textbf{C2}), and they are separated by a non-overlapping fence path~$\pi_\mathrm{fence}$.
The overlap condition is:
% $s_\mathrm{ovlp} = S(\pi_\mathrm{left}) = S(\pi_\mathrm{pump}) = S(\pi_\mathrm{loop}) = S(\pi_\mathrm{bridge})$.
$\mathrm{Ovlp}(\pi_\mathrm{left},\, \pi_\mathrm{pump},\, \pi_\mathrm{loop},\, \pi_\mathrm{bridge})$.
\end{pattern}

\noindent
\Cref{fig:patterns}(4) illustrates Pattern~2.
The fence path $\pi_\mathrm{fence}$ is critical: it separates the two loops so that they do not form a double-overlap-loop.
If $\pi_\mathrm{fence}$ were to overlap with $\pi_\mathrm{pump}$ or $\pi_\mathrm{loop}$, the two loops would constitute a classical IDA pattern detectable by existing detectors.
The non-overlapping fence is precisely what makes this pattern invisible to IDA detectors and unique to REwB.

% ---- Pattern 3 ----
\begin{pattern}{Backref-to-Loop-and-Loop}{3}
A 2PMFA contains Pattern~3 if it has a path of the form
\[
  \pi_\mathrm{prefix}
  \;\xrightarrow{\lcg_i}\;
  \pi_\mathrm{left}\;\pi_\mathrm{pump}^*\;\pi_\mathrm{fence}\;\pi_\mathrm{loop}^*\;\pi_\mathrm{right}
  \;\xrightarrow{\rcg_i}\;
  \pi_\mathrm{bridge}
  \;\xrightarrow{\backslash i}\;
  \pi_\mathrm{suffix}
\]
where $\pi_\mathrm{pump}$ and $\pi_\mathrm{loop}$ are \emph{distinct} loops that \emph{both} reside inside the capture group, separated by a non-overlapping fence $\pi_\mathrm{fence}$.
$\pi_\mathrm{pump}$ provides the non-$\mathcal{O}(1)$ captured length (\textbf{C1}), and $\pi_\mathrm{loop}$ provides the non-$\mathcal{O}(1)$ evaluation count (\textbf{C2}).
The overlap condition is:
% $s_\mathrm{ovlp} = S(\pi_\mathrm{left}) = S(\pi_\mathrm{pump}) = S(\pi_\mathrm{loop}) = S(\pi_\mathrm{right}) = S(\pi_\mathrm{bridge})$.
$\mathrm{Ovlp}(\pi_\mathrm{left},\, \pi_\mathrm{pump},\, \pi_\mathrm{loop},\, \pi_\mathrm{right} \pi_\mathrm{bridge})$.
\end{pattern}

\noindent
\Cref{fig:patterns}(5) illustrates Pattern~3, which lies structurally between the other two.
As in Pattern~2, two distinct loops are separated by a non-overlapping fence, preventing IDA.
As in Pattern~1, all relevant loops reside inside the capture group, so the backreference matches the full captured content (including both pumped portions).

Overall, three patterns evade existing IDA detectors.
Patterns~2 and~3 contain two loops but separate them with a non-overlapping fence, breaking the double-overlap-loop structure.
Pattern~1 contains only a single loop altogether.
In each case, the vulnerability arises specifically from the interaction between the loop(s) and the backreference---a mechanism that previously established runtime analyses, which assume $\mathcal{O}(1)$ per-transition cost, cannot capture.

\subsection{Super-linear Runtime Proofs} \label{sec:sufficiency}

We now prove that each pattern is sufficient to cause super-linear runtime.
For clarity, we present proof sketches for simplified versions of the patterns, in which the generalized subpaths $\pi_\mathrm{prefix}$, $\pi_\mathrm{left}$, $\pi_\mathrm{right}$, and $\pi_\mathrm{bridge}$ are instantiated as~$\epsilon$.
Full proofs are deferred to \Cref{sec:proofthm4}.

\begin{theorem}{ReWB Super-Linear Runtime}{patterns-suff}
For an $\epsilon$-loop-free 2PMFA $A$, if $A$ contains \Cref{pat:1,pat:2,pat:3}, then $\mathrm{BtRtN}(A, n) \notin \mathcal{O}(n)$.
\end{theorem}

\begin{proof}
We construct an attack string for each pattern and show that it induces $\Omega(n^2)$ backtracking runtime.

\smallskip\noindent
\textbf{Pattern~1} (single loop, \Cref{fig:patterns}(3)).
Consider the simplified path $\lcg_i\;\pi_\mathrm{pump}^*\;\rcg_i\;\backslash i\;\pi_\mathrm{suffix}$ and the attack string
\[
  s \;=\; s_\mathrm{ovlp}^{2n'} \;\; s_\mathrm{nsuffix},
  \; \text{where } S(\pi_\mathrm{pump}) = s_\mathrm{ovlp},\;
  s_\mathrm{nsuffix} \neq S(\pi_\mathrm{suffix}).
\]
The greedy loop first consumes all $2n'$ copies of $s_\mathrm{ovlp}$.
Because $s_\mathrm{nsuffix}$ forces a mismatch at $\pi_\mathrm{suffix}$, the engine backtracks, reducing the loop's consumption from $2n'$ down to~$0$.
When the loop matches $k$ copies ($k \leq n'$), the backreference $\backslash i$ attempts to re-match the captured string of length $k \cdot |s_\mathrm{ovlp}|$, costing $\mathcal{O}(k)$ time.
The total backreference cost is therefore
\[
  \sum_{k=0}^{n'} k \cdot |s_\mathrm{ovlp}|
  \;=\; |s_\mathrm{ovlp}| \cdot \frac{n'(n'+1)}{2}
  \;\in\; \Omega(n'^2).
\]
Since $n' \in \Theta(|s|)$, the runtime is $\Omega(|s|^2)$ and thus not in $\mathcal{O}(|s|)$.
% This matches the analysis of Example~1 (\S\ref{sec:example-redos-backref}), which is an instance of Pattern~1 with $s_\mathrm{ovlp} = \texttt{a}$.

\smallskip\noindent
\textbf{Pattern~2} (two separated loops, $\pi_\mathrm{loop}$ outside capture group; \Cref{fig:patterns}(4)).
Consider the simplified path $\lcg_i\;\pi_\mathrm{pump}^*\;\rcg_i\;\pi_\mathrm{fence}\;\pi_\mathrm{loop}^*\;\backslash i\;\pi_\mathrm{suffix}$ and the attack string
\[
  s \;=\; s_\mathrm{ovlp}^{n'_1} \;\; s_\mathrm{fence} \;\; s_\mathrm{ovlp}^{n'_1+n'_2} \;\; s_\mathrm{nsuffix},
\]
where $S(\pi_\mathrm{pump}) = S(\pi_\mathrm{loop}) = s_\mathrm{ovlp}$,\;
$S(\pi_\mathrm{fence}) = s_\mathrm{fence}$,\;
$s_\mathrm{nsuffix} \neq S(\pi_\mathrm{suffix})$, and $n'_1, n'_2 \in \Theta(|s|)$.
The capture group captures $s_\mathrm{ovlp}^{n'_1}$.
After matching the fence, the loop $\pi_\mathrm{loop}^*$ greedily consumes up to $n'_1 + n'_2$ copies, then backtracks. When $\pi_\mathrm{loop}^*$ matches between $n'_2$ and~$0$ copies, the backreference $\backslash i$ is evaluated $n'_2 + 1$ times, each costing $\Theta(n'_1)$ for the string comparison.
The total backreference cost is $\Omega(n'_1 \cdot n'_2)$, which is super-linear since $n'_1 \cdot n'_2 \notin \mathcal{O}(|s|)$.

\smallskip\noindent
\textbf{Pattern~3} (two separated loops, both inside capture group; \Cref{fig:patterns}(5)). The argument is analogous to Pattern~2.
The capture group now contains both $\pi_\mathrm{pump}^*$ and $\pi_\mathrm{loop}^*$ separated by $\pi_\mathrm{fence}$.
The loop $\pi_\mathrm{loop}^*$ inside the capture group still generates $\Theta(n')$ backtracking iterations, and the backreference still pays $\Theta(n')$ per evaluation for matching the captured content inflated by $\pi_\mathrm{pump}^*$, yielding $\Omega(n'^2)$ total cost.
% The full proof follows the same structure as Pattern~2 and is given in \Cref{sec:proofthm4}.
\end{proof}

\subsection{Exhaustiveness of the Classification} \label{sec:necessity-logic}

We now argue that Patterns~1--3 exhaustively cover the structural configurations implied by \Cref{thm:necessary}, under the restriction that the non-$\mathcal{O}(1)$ captured length in \textbf{C1} arises from a loop (rather than from recursive backreferences within the capture group, which we leave to future work). %; see \cref{sec:future-work}).

\Cref{thm:necessary} requires two loops to co-exist: a pump loop $\pi_\mathrm{pump}$ inside the capture group (\textbf{C1}) and an evaluation loop $\pi_\mathrm{loop}$ before or around the backreference (\textbf{C2}).
\Cref{fig:patterns} summarizes the case analysis.
Three structural decisions determine the pattern:

\begin{enumerate}[label=\textbf{D\arabic*}.]
  \item \textbf{Are $\pi_\mathrm{pump}$ and $\pi_\mathrm{loop}$ the same loop?}
  If yes, a single loop satisfies both conditions, yielding \textbf{Pattern~1}.
  If no, we proceed to \textbf{D2}.

  \item \textbf{Does $\pi_\mathrm{loop}$ reside inside or outside the capture group?}
  Since $\pi_\mathrm{pump}$ is inside the capture group (by \textbf{C1}), $\pi_\mathrm{loop}$ can be either inside or outside.
  If outside, we proceed to \textbf{D3(a)}; if inside, to \textbf{D3(b)}.

  \item \textbf{Do $\pi_\mathrm{pump}$ and $\pi_\mathrm{loop}$ overlap?}
  \begin{enumerate}[label=(\alph*)]
    \item $\pi_\mathrm{loop}$ is outside the capture group. If the two loops overlap, they form a classical IDA (double-overlap-loop) pattern, which is already detectable by existing tools and falls outside our scope ($\mathcal{O}(n)$ sink ambiguity). If they are separated by a non-overlapping fence $\pi_\mathrm{fence}$, we obtain \textbf{Pattern~2}.
    \item $\pi_\mathrm{loop}$ is inside the capture group. By the same argument, overlapping loops yield IDA. Non-overlapping loops separated by $\pi_\mathrm{fence}$ yield \textbf{Pattern~3}.
  \end{enumerate}
\end{enumerate}

\noindent
One remaining case is when \textbf{C1} is satisfied not by a loop but by a backreference nested within the capture group (\ie cycle-referencing between capture groups).
Such recursive patterns are complex, rarely encountered in practice (none appeared in our evaluation; \cref{sec:evaluation}), and their analysis involves undecidable intersection problems for 2PMFAs~\cite{carle2009,chida2023}.
We therefore leave their characterization to future work and note that our classification is exhaustive for the loop-based case, which covers all vulnerabilities found in our evaluation.

\begin{rqanswer}[RQ2 (Detection)]
We identify three structural vulnerability patterns (\Cref{pat:1,pat:2,pat:3}) derived from the necessary conditions in \Cref{thm:necessary}.
Pattern~1 uses a single loop inside the capture group; Patterns~2 and~3 use two distinct loops separated by a non-overlapping fence (outside and inside the capture group, respectively).
Each pattern evades existing IDA-based detectors yet induces $\Omega(n^2)$ runtime.
For each detected pattern, we construct an attack automaton from which adversarial inputs can be systematically extracted.
\end{rqanswer}

%-------------------------------------------------------------------------------
% \section{ReDoS Patterns and Attack Automata} % Old
% \label{sec:old_detection}
% \input{extra-text/redos_patterns_attack_automata}
%-------------------------------------------------------------------------------

%-------------------------------------------------------------------------------
\section{Evaluation}\label{sec:evaluation}
%-------------------------------------------------------------------------------

We evaluate by answering the three parts of RQ3 (\cref{sec:motivation:scope}).

%-------------------------------------------------------------------------------
\subsection{Methodology}\label{sec:eval-method}
%-------------------------------------------------------------------------------

\paragraph{Implementation.}
We implemented our detection framework by extending the Java library \texttt{dk.brics.automaton}~\cite{Moller2010Brics}, which provides NFA construction, compilation of K-regexes into NFAs, and standard NFA operations (union, intersection, minimization, emptiness checking).
Our extensions add: (1)~construction of 2PMFAs from practical REwB syntax, including two-phase memory capture and backreference evaluation; (2)~detection algorithms for Patterns~1--3 (\cref{sec:vulnerable_rewb_patterns}); and (3)~attack-automaton generators that produce adversarial inputs for each detected pattern.
We also implemented a traditional IDA detector following Wüstholz~\etal~\cite{wustholz2017} to serve as a baseline.

\paragraph{Dataset.}
We evaluate on regexes extracted from the Snort~2 registered ruleset (versions 2983--29200)~\cite{snort-rules}, a widely used network intrusion detection system.
The dataset contains 11,385 unique regexes, of which 288 (2.5\%) contain backreferences.
We excluded 10 regexes that failed to compile due to unsupported features (primarily lookaround assertions and flag modifiers) or that triggered detection errors when the tool could not compute intersections in the presence of backreferences.
This yields 278 testable REwB regexes.
\Cref{tab:detection} summarizes the dataset statistics and detection results.

\paragraph{Regex engines.}
We measure matching runtime on two production engines: PCRE~8.39 (used by Snort) and Python~3.8.10's \texttt{re} module.
Both are Spencer-style backtracking engines that support backreferences.

\paragraph{Environment.}
All experiments ran on a server with an Intel Xeon Gold 5218R (2.10\,GHz), 196\,GB RAM, and Ubuntu 20.04.6 LTS (kernel 5.4.0-216).
% Java experiments used OpenJDK 11.0.27.

%-------------------------------------------------------------------------------
\subsection{Prevalence of REwB Vulnerabilities}\label{sec:eval-detection}
%-------------------------------------------------------------------------------

\Cref{tab:detection} summarizes the detection results.
Among the 278 testable REwB regexes, our detector identifies 45 previously unknown backreference-induced ReDoS vulnerabilities---none of which are flagged by the IDA-only baseline.
All 45 match one of our three patterns; we confirmed each by manual inspection (no false positives observed).
%False negatives remain for rare regex types omitted by our analysis (\cref{sec:conclusion}).
\JD{Table 2: Bottom row: Ideally this table does not report 1337 from IDA-only, but just the results on the 278 REwB subset? That's what the top of the table does, it filters for us, so it's confusing to change the scope at the bottom.}
\begin{table}[!ht]
  \centering
  \caption{Dataset statistics and detection results. \emph{Pattern~$k$\,only}: Pattern~$k$ without co-occurring IDA\@. \emph{Pattern~$k$\,{+}\,IDA}: Pattern~$k$ co-occurring with IDA\@. \emph{IDA-only}: IDA-flagged regexes by the baseline~\cite{wustholz2017} that do \emph{not} match any of Patterns~1--3.}
  \label{tab:detection}
  \footnotesize
  \setlength{\tabcolsep}{14pt}
  \begin{tabular}{@{}l rrr@{}}
    \toprule
    \multicolumn{4}{@{}l}{\textsc{Dataset}} \\
    \cmidrule(l){1-4}
    Total regexes       & \multicolumn{3}{r@{}}{11,385} \\
    Containing backrefs & \multicolumn{3}{r@{}}{288\;\;(2.5\%)} \\
    Excluded (unsupported features) & \multicolumn{3}{r@{}}{2,129} \\
    \textbf{Tested REwB}           & \multicolumn{3}{r@{}}{\textbf{278}} \\
    \midrule
    \multicolumn{4}{@{}l}{\textsc{Detection Results (among 278 tested REwB)}} \\
    \cmidrule(l){1-4}
                        & Only  & {+}\,IDA & Total \\
    \cmidrule(lr){2-2} \cmidrule(lr){3-3} \cmidrule(l){4-4}
    Pattern 1           &  1    &  8       &  9    \\
    Pattern 2           & 14    & 22       & 36    \\
    Pattern 3           &  0    &  0       &  0    \\
    \cmidrule(l){4-4}
    \textbf{Patterns 1--3 (ours)} & \textbf{15} & \textbf{30} & \textbf{45} \\
    \midrule
    IDA-only (baseline)~\cite{wustholz2017} & \multicolumn{3}{r@{}}{1,337} \\
    \textbf{All vulnerable}       & \multicolumn{3}{r@{}}{\textbf{1,379}} \\
    \bottomrule
  \end{tabular}
\end{table}

\paragraph{Pattern distribution.}
Pattern~2 (\cref{fig:patterns}(4)) accounts for the majority of findings (36 of 45).
This is unsurprising: many Snort regexes place an ``any-character'' quantifier such as \verb|.*| before a backreference, which naturally forms the external loop $\pi_\mathrm{loop}$ required by Pattern~2.
The overlap constraint is easily satisfied because such loops accept any symbol, and a non-overlapping fence $\pi_\mathrm{fence}$ frequently separates the two loops.
Pattern~1 accounts for the remaining 9 cases.
Pattern~3, which requires two distinct loops within a single capture group, does not appear---consistent with the observation that capture groups in Snort regexes tend to be syntactically simple, typically containing at most one quantifier.

\paragraph{Co-occurrence with IDA.}
Of the 45 REwB vulnerabilities, 30 co-occur with an IDA pattern.
The remaining 15 are exclusively backreference-induced: their sink ambiguity is $O(n)$, so they are invisible to any IDA-based detector.
As we show in \cref{sec:eval-runtime}, the co-occurring cases are particularly dangerous because the two vulnerability sources compound.

\paragraph{Detection time.}
\Cref{fig:detect-time} reports detection time across all 278 regexes.
Pattern~1 is the cheapest to detect (median $<$\,0.01\,s), as it requires locating only a single loop.
Pattern~3 is faster than Pattern~2 (median 0.02\,s vs.\ 0.05\,s) because loop pairs are searched within the restricted scope of a capture group.
Pattern~2 incurs the highest overhead, with a worst case of approximately 1.5\,s for the most complex regexes.
These times are acceptable for offline auditing and CI/CD integration but may be too high for online, per-packet analysis---a tradeoff consistent with other static ReDoS detectors~\cite{wustholz2017,hassanImprovingDevelopersUnderstanding2023}.

\begin{figure}[!ht]
    \centering
    \includegraphics[width=\linewidth]{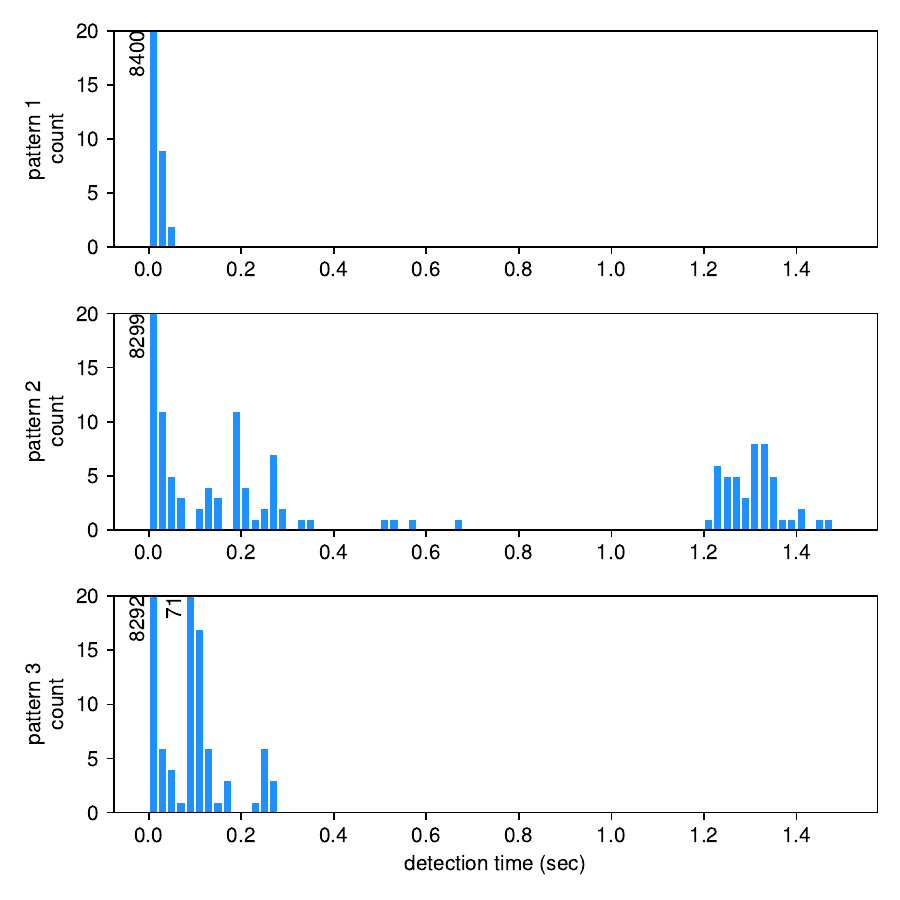}
    \caption{Static analysis time for detecting Patterns~1--3 across all 278 REwB. Most regexes are analyzed in under 0.1\,s.}
    \label{fig:detect-time}
\end{figure}

% \begin{finding}
% Our detector uncovers 45 backreference-induced ReDoS vulnerabilities in the Snort ruleset, 15 of which are invisible to existing IDA-based detectors. Pattern~2 is the most prevalent, appearing in 36 of the 45 cases.
% \end{finding}

%-------------------------------------------------------------------------------
\subsection{Runtime Impact}\label{sec:eval-runtime}
%-------------------------------------------------------------------------------

We now measure how the detected vulnerabilities manifest as runtime degradation on real engines.

\paragraph{Procedure.}
For each of the 45 vulnerable regexes, we generate three families of adversarial inputs from the corresponding attack automata: (1)~inputs exploiting only the REwB pattern (Pattern~$k$-only), (2)~inputs exploiting only the co-occurring IDA pattern (IDA-only), and (3)~inputs exploiting both simultaneously (Pattern~$k$+IDA).
For each family, we vary the pump length to produce inputs of increasing size.
Each regex--input pair is executed 10 times per engine; we report the mean wall-clock matching time.
To characterize the growth rate, we fit the measurements to a degree-4 polynomial via least-squares regression and identify the dominant term by inspecting coefficient significance.

\paragraph{Results.}
\Cref{fig:runtime} shows representative results for a regex exhibiting both Pattern~2 and IDA.
On Python~3, both the Pattern~2-only and IDA-only attacks yield quadratic runtime ($\Theta(n^2)$), consistent with our theoretical prediction.
When the attack input exploits both Pattern~2 and IDA simultaneously, runtime escalates to \emph{cubic} growth ($\Theta(n^3)$), confirming that the two vulnerability sources compound multiplicatively.
On PCRE, the IDA-only and Pattern~2-only attacks produce modest super-linear growth. The combined Pattern~2+IDA attack, however, triggers pronounced non-linear behavior, with matching times exceeding 1\,s for inputs of length~$\sim$3,000.
The remaining 44 regexes exhibit qualitatively similar trends.% ; detailed plots are provided in~\cref{app:runtime_plots}.

\begin{figure}[t]
    \centering
    \includegraphics[width=\linewidth]{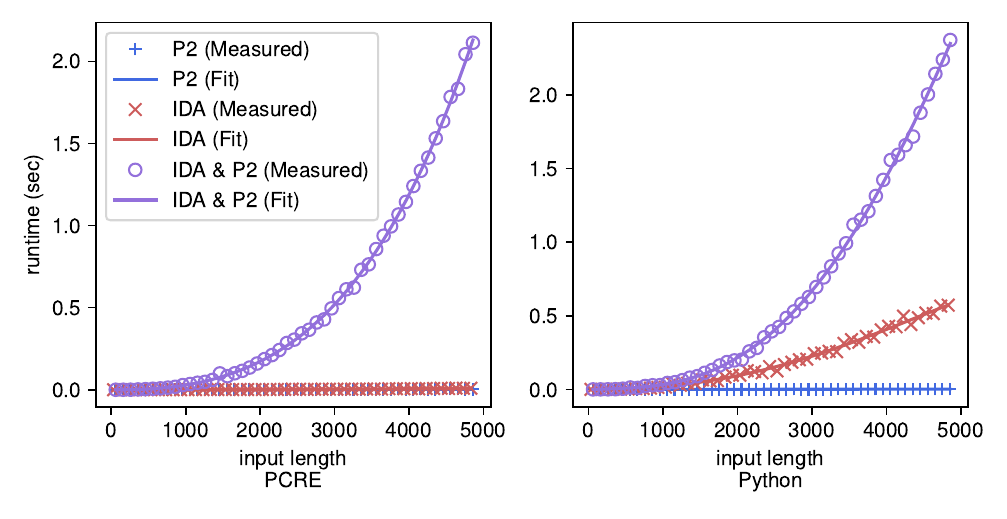}
    % \begin{exploitlst}
%(\w+)\s*=\s*(\x22npdivx\.DivXBrowserPlugin\x22|\x27npdivx\.DivXBrowserPlugin\x27)\s*\x3b.*(\w+)\s*=\s*new\s*ActiveXObject\s*\(\s*\1\s*\)(\s*\.\s*(Resize)\s*\(|.*\3\s*\.\s*(Resize)\s*\()|(\w+)\s*=\s*new\s*ActiveXObject\s*\(\s*(\x22npdivx\.DivXBrowserPlugin\x22|\x27npdivx\.DivXBrowserPlugin\x27)\s*\)(\s*\.\s*(Resize)\s*\(|.*\7\s*\.\s*(Resize)\s*\()
    %\end{exploitlst}
    \caption{Matching time on PCRE and Python for a representative Snort regex, under three attack strategies: Pattern~2 only, IDA only, and their combination.}
    \label{fig:runtime}
\end{figure}

% \begin{finding}
% Backreference-induced patterns cause quadratic runtime on their own, but when combined with IDA, runtime compounds to cubic or worse. On PCRE, combined attacks exceed 1\,s for inputs under 4,000 characters.
% \end{finding}

%-------------------------------------------------------------------------------
\subsection{Exploitability in Snort}\label{sec:eval-snort}
%-------------------------------------------------------------------------------

Finally, we assess whether the detected vulnerabilities are exploitable in a deployed system.
We target Snort~2.9, which uses PCRE for regex-based packet inspection.

\paragraph{Setup.}
We implemented a TCP client--server pair running on two separate virtual machines on the same physical host (Snort does not inspect localhost traffic).
The client sends a crafted packet; Snort, running inline between the two VMs, inspects the payload against its loaded rules.
To obtain precise timing, we instrumented the PCRE library to record per-match wall-clock time with minimal overhead.

\paragraph{Exploit strategies.}
We identified four concrete exploits (\cref{exp:1,exp:2,exp:3,exp:4}, detailed in \cref{sec:snort}) that demonstrate two distinct attack strategies.

\emph{Strategy~1: Performance degradation.}
Exploits~1 and~2 target regexes whose combined Pattern~2 and IDA structure yields $\Omega(n^3)$ matching time.
With attack strings of approximately 3,000 characters, PCRE matching takes 0.6--1.2\,seconds per packet---orders of magnitude slower than the microsecond-scale budget of a packet inspection system.
At network line rates, this is sufficient to degrade Snort's throughput or force it to drop packets.

\emph{Strategy~2: Alert bypass.}
Exploits~3 and~4 craft two-part attack packets.
The first part triggers extensive backtracking, exhausting PCRE's configurable matching limit (\texttt{pcre\_match\_limit}).
Once the limit is reached, PCRE aborts the match and Snort skips the rule.
The second part carries the actual malicious payload (\eg an ActiveX instantiation or an XSLT entity injection), which Snort no longer inspects.
This enables complete evasion of the targeted detection rule.

\noindent\textbf{\ul{Responsible disclosure:}}
All four exploits have been disclosed to the Snort development team.

\begin{rqanswer}[RQ3 (Prevalence and Impact)]
\textbf{(a)}~Among 278 testable REwB in Snort, we detect 45 backreference-induced vulnerabilities, 15 of which are invisible to IDA-based detectors.
\textbf{(b)}~Backreference patterns alone induce $\Theta(n^2)$ runtime; when combined with IDA, runtime compounds to $\Theta(n^3)$ or worse.
\textbf{(c)}~We demonstrate four exploits: two cause 0.6--1.2\,s matching delays per packet, and two bypass detection entirely by exhausting PCRE's matching limit.
\end{rqanswer}

%-------------------------------------------------------------------------------
% \section{Evaluation} % Old
% \label{sec:old_evaluation}
% \input{extra-text/evaluation}
%-------------------------------------------------------------------------------

%-------------------------------------------------------------------------------
\section{Related Work}
\label{sec:related_work}
%-------------------------------------------------------------------------------

% We situate our contributions within three lines of work: ReDoS detection techniques, theoretical complexity foundations for regex matching, and the expressiveness of backreferences.

\myparagraph{ReDoS Detection.}
Static detectors~\cite{Rathnayake2014rxxr2, Kirrage2013rxxr, weideman2016sink, wustholz2017, hassanImprovingDevelopersUnderstanding2023} model regexes as NFAs and search for structural patterns (\eg two-overlap loops) that imply super-linear backtracking. Our work falls within this category.
Dynamic tools~\cite{Shen2018ReScueGeneticRegexChecker, mclaughlin_regulator_2022, Petsios2017SlowFuzz} fuzz regex engines to find slow inputs, while hybrid approaches~\cite{liuRevealerDetectingExploiting2021, li_redoshunter_2021, wang_effective_2023} combine both paradigms.
To the best of our knowledge, all existing methods assume $O(1)$-cost transitions and operate on K-regex semantics, making them blind to the backreference-induced vulnerabilities we identify.
For a comprehensive survey, see Bhuiyan~et~al.~\cite{SoKReDoS}.

\myparagraph{Complexity Foundations.}
Weber and Seidl~\cite{weber1991ida} connect NFA ambiguity to two-overlap-loop structures, and Weideman~et~al.~\cite{weideman2016sink} show that sink ambiguity upper-bounds backtracking runtime.
These form the basis for existing detectors but assume $O(1)$-cost transitions.
Our 2PMFA extends this framework to non-$\mathcal{O}(1)$ cost transitions, and our Theorems~1--2 generalize the runtime--ambiguity relationship to REwB.

\myparagraph{Backreferences.}
Aho~\cite{ahoPatternMatchingStrings1980} proved that matching with backreferences is NP-complete.
Subsequent work studies REwB expressiveness~\cite{nogamiRegularExpressionsBackreferences2024,nogamiExpressivePowerRegular2023,Campeanu2003FormalPracticalRegexes,berglundReexaminingRegularExpressions2023} and automata models for backreference semantics~\cite{schmid2016}, but none addresses \emph{which} REwB patterns cause super-linear backtracking or \emph{how} to detect them---the questions this paper answers.

\section{Discussion and Conclusion}
\label{sec:conclusion}

This paper presents the first systematic study of ReDoS vulnerabilities caused by backreferences.
We introduced the Two-Phase Memory Finite Automaton (2PMFA) to formally analyze backreference-induced complexity, and derived necessary conditions under which REwB sees super-linear backtracking despite appearing safe to prior tooling.
From these conditions we identified three novel vulnerability patterns, developed detection and attack-generation algorithms, and uncovered 45 previously unknown vulnerabilities in the Snort intrusion detection ruleset---15 of which are invisible to existing IDA-based detectors.
We demonstrated practical exploits that degrade Snort's packet inspection by 0.6--1.2\,s or bypass detection entirely by exhausting PCRE's matching limit.

\paragraph{Limitations.}
Our pattern classification covers the loop-based case.
False negatives remain for vulnerabilities arising from cyclic backreference chains (\cref{sec:necessity-logic}) --- this scenario does not occur in the Snort data.
We evaluate on a single corpus (Snort); while Pattern~2's \verb|.*\k| idiom is common across regex-heavy applications, the prevalence of backreference patterns in other domains remains to be confirmed.

\paragraph{Implications.}
Operators of systems that evaluate regexes on untrusted input should audit backreference-containing rules with our patterns rather than relying solely on IDA-based tools.
Engine developers should consider complexity guards that account for non-constant transition costs, as tightening \texttt{pcre\_match\_limit} alone can itself become an attack vector.

\paragraph{Future work.}
Extending 2PMFA to other irregular regex features, characterizing cyclic backreference vulnerabilities, and developing semantics-preserving repair strategies for vulnerable REwB are natural next steps.

\JD{Final checklist: Run the Python tool to avoid desk reject} 

\clearpage

\section*{Ethical Considerations}
This section outlines the ethical considerations associated with our work.
The central ethical issue raised by our techniques is their applicability to vulnerability discovery, which creates a familiar ``dual-use'' context with both potential risks and benefits.
We conducted a stakeholder-based ethics analysis following the framework proposed by Davis \etal~\cite{davis2025guide}.

\subsection*{Stakeholders}

\paragraph{Direct stakeholders}
\begin{itemize}
    \item \textbf{Software engineers and maintainers.}
    Developers who apply our techniques to analyze their own regular expressions for ReDoS vulnerabilities.
    These stakeholders directly interact with the analysis outputs and decide how to respond to identified risks.
    \item \textbf{Regex engine developers and maintainers.}
    Engineers responsible for the design and implementation of regular expression engines, who may use our results to inform runtime mitigations, engine-level defenses, or design tradeoffs in their implementations of backreferences.
    \item \textbf{System operators.}
    Teams responsible for deploying and operating software systems that rely on potentially vulnerable regular expressions, including web services, infrastructure software, and embedded systems.
    \item \textbf{Adversaries.}
    Malicious actors who could adopt the techniques described in this paper to identify denial-of-service exploits targeting regular expressions.
    \item \textbf{The research team.}
    The authors of this work.
\end{itemize}

\paragraph{Indirect stakeholders}
\begin{itemize}
    \item \textbf{End users of affected software.}
    Individuals or organizations that depend on systems incorporating vulnerable regular expressions and may experience service degradation or outages due to exploitation.
    \item \textbf{The broader software ecosystem.}
    Maintainers and users of libraries, frameworks, and applications that embed or reuse regular expressions with backreferences.
    \item \textbf{Vulnerable or high-impact user groups.}
    Populations that may be disproportionately harmed by denial-of-service attacks, including users of safety-critical, medical, industrial, or civic software systems.
    \item \textbf{The security research community.}
    Researchers and practitioners who may build upon, extend, or operationalize the techniques presented in this work.
\end{itemize}

\subsection*{Potential Harms and Mitigating Factors}

\begin{itemize}
    \item \textbf{Facilitation of exploitation.}
    Our techniques could lower the cost for attackers to identify or construct denial-of-service attacks involving regular expressions with backreferences, potentially enabling exploitation.
%    \item \textbf{False positives and engineering burden.}
%    Analysis results may include false positives that consume developer time or lead to unnecessary code changes.
%    To address this risk, our evaluation characterizes the prevalence and nature of false positives to support informed adoption.
    \item \textbf{Operational and economic costs.}
    Mitigating identified vulnerabilities may require refactoring, performance tradeoffs, or service interruptions, imposing costs on organizations and operators.
    \item \textbf{Overconfidence or misuse.}
    Developers may incorrectly interpret our techniques as providing comprehensive protection against all forms of denial-of-service or input-related vulnerabilities.
    To mitigate this risk, we explicitly delimit the classes of vulnerabilities addressed by our methods, using theorems and proofs.
    \item \textbf{Risk to researchers.}
    The research team may face reputational or legal exposure if the techniques are misapplied or framed as enabling harmful activity.
\end{itemize}

\subsection*{Potential Benefits}

\begin{itemize}
    \item \textbf{Improved robustness against ReDoS.}
    Our techniques enable earlier and more systematic identification of denial-of-service risks arising from regular expressions with backreferences.
    \item \textbf{Support for defensive engineering practices.}
    By providing structured analyses of problematic regular expressions, this work can inform both application-level remediation and engine-level mitigation strategies.
    \item \textbf{Guidance for regex engine design.}
    Empirical evidence about the behavior of backreferences and pathological patterns can help engine maintainers reason about performance safeguards and runtime defenses.
    \item \textbf{Advancement of theoretically-grounded security research.}
    This work contributes to the understanding of denial-of-service vulnerabilities in regular expression engines, enabling further defensive research and tooling.
\end{itemize}

\subsection*{Judgment}

In our assessment, the anticipated benefits to software security outweigh the risks associated with this work.
We considered the ethical implications of our techniques from the outset of the study.
No additional ethical concerns emerged during the course of the research.
On this basis, we proceeded with submission to USENIX.

\section*{Open Science}
Anonymized artifacts accompany our submission. They are available at \url{https://anonymous.4open.science/r/slmad-EABE} and
\url{https://anonymous.4open.science/r/atkre-7D50}.
%The first repository is the code of detector, while the second one is mainly for the dynamic measurement of runtime.
%\url{https://anonymous.4open.science/r/redos-rewb-E365}.

The artifact includes:

\begin{enumerate}
\item \textbf{Detector:} The first repository contains the complete implementation of our detector. It takes regular expressions as input and analyzes whether a given regex follows the IDA pattern or one of the three non-IDA patterns. If so, it generates corresponding attack strings.
\item \textbf{Dynamic Validation:} The second repository focuses on dynamic runtime measurement. It invokes different regex engines to match the regexes against the attack strings generated by the detector, measures the execution time, and fits the relationship between input length and runtime using a polynomial curve.
\item \textbf{Plotting Scripts:} The scripts in the \texttt{plot} directory of the second repository are used to generate the plots presented in this paper.
\item \textbf{Data:} We provide all input, output, and intermediate datasets. The first repository includes the input regexes, their pattern classifications, and the corresponding potential attack strings. The second repository contains the measured runtimes and the fitted curves.
\end{enumerate}

%\begin{enumerate}
%\item \textbf{Implementation.}
%The reference implementation of our analyzer for regular expressions with backreferences, including all components required to reproduce the analyses described in the paper.
%\item \textbf{Measurement scripts.}
%Scripts used to run experiments, collect measurements, and process intermediate results.
%\item \textbf{Data.}
%The datasets produced and consumed by the measurement scripts, sufficient to regenerate all tables and figures reported in the manuscript.
%\end{enumerate}

Taken together, the artifact contains the materials necessary to understand, inspect, and reproduce the theoretical and empirical results presented in this work.

Additional details on artifact structure, usage, and assumptions are provided in the accompanying README.

\clearpage

\bibliographystyle{plain}
\bibliography{main.bib}

\appendix

\section*{Outline of Appendices}

\noindent
The appendix contains the following material:

% Dongyoon, we have space issues to fit on pages 1-13. 

\begin{itemize}

\item \cref{app:self_bref}: Illustration of the semantics of self-backreferences.
\item \cref{app:matching_algorithm}: Semantics of our model of REwB using our two-phase MFA (2PMFA).
\item \cref{sec:proofthm1}: Proof of Theorem 1.
\item \cref{sec:proofthm4}: Proof of Theorem 3.
\item \cref{sec:snort}: Details of vulnerabilities found in Snort.

\end{itemize}

\section{Self-Backreference Walkthrough}
\label{app:self_bref}
\Cref{fig:self_bref} illustrates the matching behavior of self-backreferences for the regex $\lcg_1 \backslash 1 \texttt{b} \mid \texttt{a} \rcg_1^*$ on input \texttt{`aababb'}. The capture table (right column) stores the most recently committed substring for each group. Initially, group~1 is empty ($\emptyset$), so $\backslash 1$ fails and \texttt{`a'} is matched via the right branch~\circled{1}.
In subsequent iterations, $\backslash 1$ matches the previously captured substring: \texttt{`a'} at step~\circled{2}, then \texttt{`ab'} at step~\circled{3}, and so on, enabling the pattern to match progressively longer substrings.

\section{2PMFA Matching Algorithm}
\label{app:matching_algorithm}

\Cref{algo:BtRun} defines the backtracking-based matching
algorithm $\mathrm{BtRun}(A, s)$ for a 2PMFA $A$ on input
string $s$. The algorithm maintains a memory function
$M : \{\lcg'_i, \lcg_i, \rcg_i \mid i \in I\} \to
\mathbb{N}_{0..|s|} \cup \{\bot\}$ that implements the
two-phase capture group table using start and end indices
into~$s$. All entries are initially~$\bot$ (unset). We write
$f_1 \lhd f_2$ for the function that agrees with $f_2$ on its
domain and falls back to $f_1$ elsewhere (\ie
$f_1 \lhd f_2 = f_2 \cup \{x \mapsto y \mid f_1(x)=y \land
f_2(x)=\bot\}$).

The five transition rules operate as follows:

\begin{enumerate}[nosep,leftmargin=*,label=\textbf{(\roman*)}]
    \item \textbf{Symbol} ($t = \sigma \in \varSigma$):
    Consumes $s[j]$ if it equals $\sigma$, advancing the index
    by one.
    \item \textbf{Epsilon} ($t = \varepsilon$): Moves to a
    successor state without consuming input.
    \item \textbf{Open group} ($t = \lcg_i$): Records the
    current input position $j$ in the in-progress slot
    $\lcg'_i$, beginning a new capture for group~$i$.
    \item \textbf{Close group} ($t = \rcg_i$): Commits the
    capture by copying the in-progress start position into
    $\lcg_i$ and recording the current position in $\rcg_i$.
    After this step, $M(\lcg_i)..M(\rcg_i)$ delimits the most
    recently committed substring for group~$i$.
    \item \textbf{Backreference} ($t = \backslash i$): Invokes
    the helper $\mathrm{MtBr}(s, j, M, i)$, which compares
    $s[j \,..\!_< j{+}l]$ against the committed capture
    $s[M(\lcg_i) \,..\!_< M(\rcg_i)]$ where
    $l = M(\rcg_i) - M(\lcg_i)$. On success, the index
    advances by~$l$. This comparison takes $O(l)$
    time---crucially, $l$ can be as large as $O(n)$, so a
    single backreference transition is not $O(1)$.
\end{enumerate}

\paragraph{Self-reference semantics.}
When $\backslash i$ is encountered before group~$i$ has ever
been closed (\ie $M(\rcg_i) = \bot$), the behavior depends on
the engine's semantics. Under \emph{$\emptyset$-semantics}
(the default in PCRE, Python, and Java), the match fails
immediately. Under \emph{$\varepsilon$-semantics}, the
uninitialized capture is treated as the empty string, so the
backreference trivially succeeds. In~\Cref{algo:BtRun}, the
Boolean flag $b_{\mathrm{MtBrE}}$ selects between these two
behaviors.

\begin{algorithm}[t]
    \caption{Backtracking matching for 2PMFA.}
    \label{algo:BtRun}
    \begin{algorithmic}[1]
    {\small
        \REQUIRE 2PMFA $A = (Q, \varSigma, I, \varDelta, q_0, F)$,\;
                 input $s \in \varSigma^*$
        \vspace{2pt}
        \STATE $\mathrm{BtRun}(A, s) \;=\; \mathrm{BtRun}'(A, s, q_0, 0, M_\bot)$
        \FUNCTION{$\mathrm{BtRun}'(A, s, q, j, M)$}
            \IF{$q \in F \;\land\; j = |s|$}
                \RETURN \TRUE
            \ENDIF
            \FOR{each $(q, t, q') \in \varDelta$}
                \STATE $\mathit{result} \gets \FALSE$
                \SWITCH{$t$}
                    \CASE{$\sigma \in \varSigma$ \AND $j < |s|$ \AND $s[j] = \sigma$}
                        \STATE $\mathit{result} \gets \mathrm{BtRun}'(A, s, q', j{+}1, M)$
                    \ENDCASE
                    \CASE{$\varepsilon$}
                        \STATE $\mathit{result} \gets \mathrm{BtRun}'(A, s, q', j, M)$
                    \ENDCASE
                    \CASE{$\lcg_i$}
                        \STATE $\mathit{result} \gets \mathrm{BtRun}'(A, s, q', j, \;M \lhd \{\lcg'_i \mapsto j\})$
                    \ENDCASE
                    \CASE{$\rcg_i$}
                        \STATE $\mathit{result} \gets \mathrm{BtRun}'(A, s, q', j, \;M \lhd \{\lcg_i \mapsto M(\lcg'_i),\; \rcg_i \mapsto j\})$
                    \ENDCASE
                    \CASE{$\backslash i$ \AND $\mathrm{MtBr}(s, j, M, i)$}
                        \STATE $l \gets M(\rcg_i) - M(\lcg_i)$
                        \STATE $\mathit{result} \gets \mathrm{BtRun}'(A, s, q', j{+}l, M)$
                    \ENDCASE
                \ENDSWITCH
                \IF{$\mathit{result}$} \RETURN \TRUE \ENDIF
            \ENDFOR
            \RETURN \FALSE
        \ENDFUNCTION
        \vspace{4pt}
        \FUNCTION{$\mathrm{MtBr}(s, j, M, i)$}
            \IF{$M(\rcg_i) = \bot$}
                \RETURN $b_{\mathrm{MtBrE}}$
                \hfill\COMMENT{self-ref: $\emptyset$- vs.\ $\varepsilon$-semantics}
            \ENDIF
            \STATE $l \gets M(\rcg_i) - M(\lcg_i)$
            \RETURN $l \le |s| - j \;\;\land\;\; s[j \,..\!_< j{+}l] = s[M(\lcg_i) \,..\!_< M(\rcg_i)]$
        \ENDFUNCTION
    }
    \end{algorithmic}
\end{algorithm}

\section{Proof of \Cref{thm:safe-backref}}
\label{sec:proofthm1}
\begin{proof}
    %First, let us write down the algorithms for computing sink ambiguity and backtracking runtime.
    \Cref{algo:SinkAbgS} presents the algorithm for computing sink ambiguity ($\mathrm{SinkAbgS}$).
    Recall that the degree of ambiguity counts the number of accepting paths, and that a sink automaton adds an $\epsilon$-transition from every state in $Q$ to a new accepting state $q_\mathrm{sink}$ (\Cref{subsec:background_redos_and_regex_complexity}).
    %(\Cref{sec:background:sink}).
    \Cref{algo:SinkAbgS} uses nested summations (double Sigma notation) to aggregate all possible ways of reaching the accepting state $q_\mathrm{sink}$ from each state.
    
    \Cref{algo:BtRtS} presents the algorithm for computing backtracking runtime ($\mathrm{BtRtS}$). For all transitions except backreferences, the runtime variable $\tau$ is incremented by a constant (Lines \ref{algoline:BtRtS':tau:=tau+tau'+1:symbol}, \ref{algoline:BtRtS':tau:=tau+tau'+1:eps}, \ref{algoline:BtRtS':tau:=tau+tau'+1:open}, \ref{algoline:BtRtS':tau:=tau+tau'+1:close}). For a backreference transition, however, $\tau$ is incremented by the length of the captured substring (Line \ref{algoline:BtRtS':tau:=tau+tau'+1+l}).

    We now aim to scale up (approximate) $\mathrm{BtRtS}(A, s)$ so that it becomes a constant multiple of $\mathrm{SinkAbgS}(A, s)$.
    The scaled version is denoted $\mathrm{BtRtS}\!\uparrow$.
    In other words, we seek to construct $\mathrm{BtRtS}\!\uparrow$ such that there exists a constant $\xi$ satisfying:
    \[\mathrm{BtRtS}(A,s) \leq \mathrm{BtRtS}\!\uparrow(A,s) \leq \xi \cdot \mathrm{SinkAbgS}(A, s)\]
    
    %Therefore, $\mathrm{BtRtN}(A, n) \in \mathcal{O}(\mathrm{SinkAbgN}(A, n))$
    
    First, $\mathrm{BtRtS}$ currently returns a boolean indicating whether $A$ accepts $s$, and it stops early upon acceptance. We can remove this early-stopping behavior to scale it up. Specifically, the \textbf{if} statements at Lines \ref{algoline:BtRtS':if-q-in-F-and-j=s} and \ref{algoline:BtRtS':if-alpha'} can be deleted, and the boolean return can be omitted.

    Second, we move the constant additions out of the loops. To begin, we consider only backreferences that can match strings of maximum length $\mathcal{O}(1)$. This corresponds to the first constraint of \Cref{thm:safe-backref}: ``a backreference captures a string of length $\mathcal{O}(1)$ ...''.
    Let $\mathrm{MaxFBrL}(A)$ denote the maximum finite backreference length among all backreferences, which is $\mathcal{O}(1)$. Formally, 
    \begin{align*}
       &\mathrm{MaxFBrL}(A) = 1 + \\
       &\max_{\text{backref } \delta \in \varDelta \ \land \ \delta \text{ match } \mathcal{O}(1) \text{ length}} \text{length that } \delta \text{ can match} 
    \end{align*}

    \begin{algorithm}[h!] 
        \caption{Sink Ambiguity w.r.t. string ($\mathrm{SinkAbgS}$)}
        \label{algo:SinkAbgS}
        \begin{algorithmic}
        {\small
            \REQUIRE An 2PMFA $A = (Q, \varSigma, I, \varDelta, q_0, F)$
            \REQUIRE A string $s \in \varSigma^*$ 
            \REQUIRE A current state $q \in Q$
            \REQUIRE An index $j \in \mathbb{N}_{0..|s|}$ of $s$
            \REQUIRE A memory function $M : \{\lcg_i', \lcg_i, \rcg_i \mid i \in I\} \to \mathbb{N}_{0..|s|}$
        }
        \end{algorithmic}
        {\footnotesize
        \begin{flalign*}
            &\mathrm{SinkAbgS}(A, s) = \mathrm{SinkAbgS}'(A, s, q_0, 0, \emptyset) \\
            &\mathrm{SinkAbgS}'(A, s, q, j, M) = 1 + \sum_{((q, t) \mapsto Q') \in \varDelta} \ \sum_{q' \in Q'} \\
            & \begin{cases}
                \begin{cases}
                    \mathrm{SinkAbgS}'(A, s, q', j + 1, M) \qquad\qquad & j < |s| \land s[j] = t \\
                    0 & \text{otherwise}
                \end{cases} & t \! \in \! \varSigma \\
                \mathrm{SinkAbgS}'(A, s, q', j, M) & t = \epsilon \\
                \mathrm{SinkAbgS}'(A, s, q', j, M \lhd \{\lcg_i \mapsto j\}) & t \text{ is } \lcg_i \\
                \mathrm{SinkAbgS}'(A, s, q', j, M \lhd \{\lcg_i \mapsto M(\lcg_i'), \; \rcg_i \mapsto j \}) & t \text{ is } \rcg_i \\
                \begin{cases}
                    \mathrm{SinkAbgS}'(A, s, q', j \! + \! M(\rcg_i) \! - \! M(\lcg_i), M) & \mathrm{MtBr}(s, j, M, i) \\
                    0 & \text{otherwise}
                \end{cases} & t \text{ is} \backslash i
            \end{cases}
        \end{flalign*}
        }
    \end{algorithm}
    
    \begin{algorithm}[h!]
    {\fontsize{8.5}{9.5}\selectfont
        \caption{Backtracking Runtime w.r.t. string ($\mathrm{BtRtS}$)}
        \label{algo:BtRtS}
        \begin{algorithmic}[1]
            \ALGOSIG $\mathrm{BtRtS}(A, s)$
            \REQUIRE An 2PMFA $A = (Q, \varSigma, I, \varDelta, q_0, F)$
            \REQUIRE A string $s \in \varSigma^*$
            \STATE $(\tau, \alpha) := \mathrm{BtRtS}'(A, s, q_0, 0, \emptyset)$
            \RETURN $\tau$
        \end{algorithmic}
        \algohline
        \begin{algorithmic}[1]
            \ALGOSIG $\mathrm{BtRtS}'(A, s, q, j, M)$
            \REQUIRE An 2PMFA $A = (Q, \varSigma, I, \varDelta, q_0, F)$
            \REQUIRE A string $s \in \varSigma^*$
            \REQUIRE A current state $q \in Q$
            \REQUIRE An index $j \in \mathbb{N}_{0..|s|}$ of $s$
            \REQUIRE A memory function $M : \{\lcg'_i, \lcg_i, \rcg_i \mid i \in I\} \to \mathbb{N}_{0..|s|}$
            \STATE $\tau := 0$    \label{algoline:BtRtS':tau=0}
            \FOR{$((q, t) \mapsto Q') \in \varDelta$}    \label{algoline:BtRtS':for-((q,t)->Q')-in-Delta}
                \IF{$q \in F \land j = |s|$}    \label{algoline:BtRtS':if-q-in-F-and-j=s}
                    \RETURN $(\tau, \mathrm{true})$
                \ENDIF    \label{algoline:BtRtS':4}
                \FOR{$q' \in Q'$}
                    \SWITCH{$t$}
                        \CASE{$t \in \varSigma$}
                            \STATE $(\tau', \alpha') := \begin{cases}
                                \mathrm{BtRtS}'(A, s, q', j+1, M) & j < |s| \land s[j] = t \\
                                (0, \mathrm{false}) & \text{otherwise}
                            \end{cases}$    \label{algoline:BtRtS':recursion:symbol}
                            \STATE $\tau := \tau + \tau' + 1$    \label{algoline:BtRtS':tau:=tau+tau'+1:symbol}
                        \ENDCASE
                        \CASE{$\epsilon$}
                            \STATE $(\tau', \alpha') := \mathrm{BtRtS}'(A, s, q', j, M)$    \label{algoline:BtRtS':recursion:eps}
                            \STATE $\tau := \tau + \tau' + 1$    \label{algoline:BtRtS':tau:=tau+tau'+1:eps}
                        \ENDCASE
                        \CASE{$\lcg_i$}
                            \STATE $(\tau', \alpha') := \mathrm{BtRtS}'(A, s, q', j, M \lhd \{\lcg'_i \mapsto j\})$    \label{algoline:BtRtS':recursion:open}
                            \STATE $\tau := \tau + \tau' + 1$    \label{algoline:BtRtS':tau:=tau+tau'+1:open}
                        \ENDCASE
                        \CASE{$\rcg_i$}
                            \STATE $(\tau', \alpha') := \mathrm{BtRtS}'(A, s, q', j, M \lhd \{\lcg_i \mapsto M(\lcg'_i), \; \rcg_i \mapsto j\})$    \label{algoline:BtRtS':recursion:close}
                            \STATE $\tau := \tau + \tau' + 1$    \label{algoline:BtRtS':tau:=tau+tau'+1:close}
                        \ENDCASE
                        \CASE{$\backslash i$}
                            \STATE $l := \begin{cases}
                                M(\rcg_i) - M(\lcg_i) & M(\rcg_i) \neq \bot \\
                                0 & \text{otherwise}
                            \end{cases}$
                            \STATE $(\tau', \alpha') := \begin{cases}
                                \mathrm{BtRtS}'(A, s, q', j + l, M) & \mathrm{MtBr}(s, j, M, i) \\
                                (0, \mathrm{false}) & \text{otherwise}
                            \end{cases}$    \label{algoline:BtRtS':recursion:backref}
                            \STATE $\tau := \tau + \tau' + 1 + l$    \label{algoline:BtRtS':tau:=tau+tau'+1+l}
                        \ENDCASE
                    \ENDSWITCH
                    \IF{$\alpha'$}    \label{algoline:BtRtS':if-alpha'}
                        \RETURN $(\tau, \mathrm{true})$
                    \ENDIF
                \ENDFOR    \label{algoline:BtRtS':endfor-q'-in'-Q'}
            \ENDFOR
            \RETURN $(\tau, \mathrm{false})$
        \end{algorithmic}
    }
    \end{algorithm}
    
    We can scale $\mathrm{BtRtS}'$ by replacing each $1 + l$ with $\mathrm{MaxFBrL}(A)$.
    We can further scale up by removing all ``$\tau := \tau + \tau' + \cdots$'' statements at Line \ref{algoline:BtRtS':tau:=tau+tau'+1:symbol}, \ref{algoline:BtRtS':tau:=tau+tau'+1:eps}, \ref{algoline:BtRtS':tau:=tau+tau'+1:open}, \ref{algoline:BtRtS':tau:=tau+tau'+1:close}, \ref{algoline:BtRtS':tau:=tau+tau'+1+l}, and putting a statement ``$\tau := \tau + \tau' + \mathrm{MaxFBrL}(A)$'' called $E_1$ at the end of ``\textbf{for} $q \in Q$'' loop (just above the Line \ref{algoline:BtRtS':endfor-q'-in'-Q'}). 
    
    After adding $E_1$, in each call to $\mathrm{BtRtS}'$, $E_1$ may be evaluated up to the maximum number of ways to transition from a given current state to other states.
    Let $\mathrm{MaxOut}(A)$ denote the maximum number of outgoing transitions (edges) across all states in $Q$, formally:
    \[
        \mathrm{MaxOut}(A) = \max_{q \in Q} \sum_{((q, t) \mapsto Q') \in \varDelta} |Q'|
    \]
    We can further scale up by replacing $E_1$ with ``$\tau := \tau + \tau'$'', and replacing the ``$\tau := 0$'' at Line \ref{algoline:BtRtS':tau=0} with ``$\tau := \mathrm{MaxOut}(A) \cdot \mathrm{MaxFBrL}(A)$''.  Note that $\mathrm{MaxFBrL}(A), \mathrm{MaxOut}(A) \in \mathcal{O}(1)$ with respect to $|s|$. 

    Next, we consider backreferences that can match strings of non-$\mathcal{O}(1)$ length.
    The second constraint of \Cref{thm:safe-backref} requires that ``these backreferences are evaluated a total of $\mathcal{O}(1)$ times.''
    %As the precondition required, each of them can evaluate totally $\mathcal{O}(1)$ counts during the whole matching. 
    Let $\varDelta_\mathrm{IBr} = \{ \delta \mid \text{backref } \delta \in \varDelta \land \delta \text{ matches string of non-}O(1) \text{ length} \}$.
    Define $\mathrm{IBrRCt}(A)$ as the maximum total evaluation count of infinite backreferences. Formally,
    \begin{align*}
        \mathrm{IBrRCt}(A) = {} |\varDelta_\mathrm{IBr}| \cdot \max_{\delta \in \varDelta_\mathrm{IBr}} & \text{total number of evaluation of } \\
        & \delta \text{ during one matching} 
    \end{align*}
    Note that $\mathrm{IBrRCt}(A) \in \mathcal{O}(1)$ with respect to $|s|$.
    
    In each evaluation, a backreference transition can match a string of length at most $\mathcal{O}(|s|)$ (i.e., up to the entire input string). Instead of computing the time consumed by backreferences in $\varDelta_\mathrm{IBr}$ recursively, we directly add the scaled-up time, $\mathrm{IBrRCt}(A) \cdot |s|$, to the result in $\mathrm{BtRtS}$.

    \begin{algorithm}[h]
        \caption{Backtracking Runtime w.r.t. string, Scaled Up ($\mathrm{BtRtS}\!\uparrow$)}
        \label{algo:BtRtU}
        %\begin{algorithmic}
        %    \REQUIRE $A = (Q, \varSigma, I, \varDelta, q_0, F)$ is the 2PMFA
        %    \REQUIRE $s \in \varSigma^*$ is the string 
        %\end{algorithmic}
        %\[ 
        %    \mathrm{BtRtS}\!\uparrow(A, s) = \mathrm{BtRtS}\!\uparrow'(A, s, q_0, 0, \emptyset) + \mathrm{IBrRCt}(A) \cdot |s|
        %\]
        %\algohline
        \begin{algorithmic}
        {\small
            \REQUIRE An 2PMFA $A = (Q, \varSigma, I, \varDelta, q_0, F)$
            \REQUIRE A string $s \in \varSigma^*$ 
            \REQUIRE A current state $q \in Q$ 
            \REQUIRE An index $j \in \mathbb{N}_{0..|s|}$ of $s$
            \REQUIRE A memory function $M : \{\lcg'_i, \lcg_i, \rcg_i \mid i \in I\} \to \mathbb{N}_{0..|s|}$ 
        }
        \end{algorithmic}
        {\footnotesize
        \begin{flalign*}
            &\mathrm{BtRtS}\!\uparrow(A, s) = \mathrm{BtRtS}\!\uparrow'(A, s, q_0, 0, \emptyset) + \mathrm{IBrRCt}(A) \cdot |s| \\
            &\mathrm{BtRtS}\!\uparrow'(A, s, q, j, M) = \\
            &\mathrm{MaxOut}(A) \cdot \mathrm{MaxFRefL}(A) + \sum_{((q, t) \mapsto Q') \in \varDelta} \ \sum_{q' \in Q'} \\
            &\begin{cases}
                \begin{cases}
                    \mathrm{BtRtS}\!\uparrow'(A, s, q', j + 1, M) \qquad\qquad & j < |s| \land s[j] = t \\
                    0 & \text{otherwise}
                \end{cases} & t \! \in \! \varSigma \\
                \mathrm{BtRtS}\!\uparrow'(A, s, q', j, M) & t = \epsilon \\
                \mathrm{BtRtS}\!\uparrow'(A, s, q', j, M \lhd \{\lcg_i \mapsto j\}) & t \text{ is } \lcg_i \\
                \mathrm{BtRtS}\!\uparrow'(A, s, q', j, M \lhd \{\lcg_i \mapsto M(\lcg_i'), \; \rcg_i \mapsto j \}) & t \text{ is } \rcg_i \\
                \begin{cases}
                    \mathrm{BtRtS}\!\uparrow'(A, s, q', j \! + \! M(\rcg_i) \! - \! M(\lcg_i), M) & \mathrm{MtBr}(s, j, M, i) \\
                    0 & \text{otherise}
                \end{cases} & t \text{ is} \backslash i
            \end{cases}\\
        \end{flalign*}
        }
    \end{algorithm}

    \Cref{algo:BtRtU} presents the scaled-up version of backtracking runtime, denoted $\mathrm{BtRtS}\!\uparrow$, with respect to a string $s$.
    By comparing $\mathrm{SinkAbgS}'$ and $\mathrm{BtRtS}\!\uparrow'$, we observe that they are structurally identical, differing only by constant factors. Therefore,
    \begin{align*}
        &\mathrm{BtRtS}\!\uparrow'(A, s, q, j, M) = \\
        &\quad \mathrm{MaxOut}(A) \cdot \mathrm{MaxFBrL}(A) \cdot \mathrm{SinkAbgS}'(A, s, q, j, M)
    \end{align*}
    Plugging this into $\mathrm{BtRtS}\!\uparrow$ in Algorithm \ref{algo:BtRtU}, we obtain
    \begin{align*}
        &\mathrm{BtRtS}\!\uparrow(A, s) = \\
        &\quad \mathrm{MaxOut}(A) \cdot \mathrm{MaxFBrL}(A) \cdot \mathrm{SinkAbgS}'(A, s, q_0, 0, \emptyset) \\
        &\quad {} + \mathrm{IBrRCt}(A) \cdot |s|
    \end{align*}
    Substituting $\mathrm{SinkAbgS}$ from Algorithm \ref{algo:SinkAbgS} gives
    \begin{align*}
        &\mathrm{BtRtS}\uparrow(A,s) = \\
        &\quad \mathrm{MaxOut}(A) \cdot \mathrm{MaxFBrL}(A) \cdot \mathrm{SinkAbgS}(A,s) + {} \\
        &\quad \mathrm{IBrRCt}(A) \cdot |s|
    \end{align*}
    Since $\mathrm{BtRtS}(A,s) \leq \mathrm{BtRtS}\uparrow(A,s)$, we have
    \begin{align*}
        &\mathrm{BtRtS}(A,s) \leq \\
        &\quad \mathrm{MaxOut}(A) \cdot \mathrm{MaxFBrL}(A) \cdot \mathrm{SinkAbgS}(A,s) + {} \\
        &\quad \mathrm{IBrRCt}(A) \cdot |s|
    \end{align*}
    
    %We now compare $\mathrm{SinkAbgS}$ with $\mathrm{BtRtS}!\uparrow$ across three cases.
    We now consider three possible cases for $\mathrm{SinkAbgS}$ and show that there exists a constant $\xi$ such that
    \[\mathrm{BtRtS}(A,s) \leq \xi \cdot \mathrm{SinkAbgS}(A, s)\]

    \begin{description}[leftmargin=0pt,itemsep=5pt]
        \item[Case 1] $\mathrm{SinkAbgS}(A, s) \in \varOmega(|s|)$.
        %; that is, $\mathrm{SinkAbgS}$ grows super-linearly.
        %%
        Since $\mathrm{IBrRCt}(A) \cdot |s| \in \mathcal{O}(|s|)$, adding it to an $\varOmega(|s|)$ term does not change the asymptotic scale. Therefore, $\exists \xi : \mathrm{BtRtS}(A,s) \leq \xi \cdot \mathrm{SinkAbgS}(A,s)$.

        \item[Case 2] $\mathrm{SinkAbgS}(A, s) \in \mathcal{O}(1)$ (constant).
        We argue by contradiction that, in this case, $A$ cannot contain any reachable loop.
        If a reachable loop existed, it would combine with the sink loop to create a double-overlap-loop structure, causing the sink automaton to exhibit IDA behavior; that is, $\mathrm{SinkAbgN}(A, n) \notin \mathcal{O}(1)$.
        This contradicts the assumption of this case.

        Because $A$ has no reachable loops, there can be no looping capture groups or cycles involving backreferences. Consequently, no capture group can match a string of non-$\mathcal{O}(1)$ length, and the same holds for backreferences. Thus $|\varDelta_\mathrm{IBr}| = 0$, and therefore $\mathrm{IBrRCt}(A) = 0$.
        Hence, $\exists \xi : \mathrm{BtRtS}(A,s) \leq \xi \cdot \mathrm{SinkAbgS}(A,s)$.
        
        \item[Case 3] $\mathrm{SinkAbgS}(A, s) \in \mathrm{Complement}(\varOmega(|s|)) \setminus \mathcal{O}(1)$.
        In other word, $\mathrm{SinkAbgS}(A, s)$ is less than $\mathcal{O}(n)$ but greater than $\varOmega(1)$.

        In this case, we first show by contradiction that $A$ must contain a reachable loop. If $A$ had no reachable loops, then along any path $\pi$ of $A$, each transition $\delta \in \varDelta$ could appear at most once. The total number of possible paths from $q_0$ to any state $q \in Q$ would then be bounded by $\sum_{k=0}^{|\varDelta|} \mathrm{P}_{|\varDelta|}^k$, which is a constant with respect to $|s|$. In the sink automaton $\mathrm{Sink}(A)$, every state has an $\epsilon$-transition to the sink state, so the number of paths from $q_0$ to each state equals the number of paths from $q_0$ to the sink state, i.e., the degree of sink ambiguity. This would imply that the sink ambiguity is $\mathcal{O}(1)$ in $|s|$, contradicting the assumption of this case.

        Next, we show, again by contradiction, that every reachable loop in $A$ must contain a backreference that matches a non-$\mathcal{O}(1)$-length string. Suppose there exists a reachable loop consisting only of the following types of transitions: symbol, $\epsilon$, capture-open, capture-close, or backreferences that match only $\mathcal{O}(1)$-length strings. Then, the double-overlap structure created by such a loop together with the sink loop would yield at least $\varOmega(|s|)$ sink ambiguity, again contradicting the assumption.
        %of sublinear (but non-constant) $\mathrm{SinkAbgS}(A,s)$.

        Because no reachable loop can be formed solely from $\mathcal{O}(1)$-transitions, at least one backreference in $A$ must match strings of non-$\mathcal{O}(1)$ length via cyclic referencing. Such a backreference is therefore evaluated a non-$\mathcal{O}(1)$ number of times and matches a non-$\mathcal{O}(1)$-length substring. This violates the theorem’s precondition, so this entire case does not need to be considered.
    
    \end{description}

    Considering the two valid cases above, we obtain
    \begin{align*}
        &\exists \xi :  \forall n \in \mathbb{N} : \exists \ s \in \varSigma^n : \\
        &\quad \begin{aligned}
            \mathrm{BtRtN}(A, n) &= \mathrm{BtRtS}(A, s) \\
            &\leq \xi \cdot \mathrm{SinkAbgS}(A, s) \\
            &\leq \xi \cdot \mathrm{SinkAbgN}(A, n)
        \end{aligned}
    \end{align*}

    Therefore, $\mathrm{BtRtN}(A, n) \in \mathcal{O}(\mathrm{SinkAbgN}(A, n))$.
    % ′
\end{proof}

\section{Proof of \Cref{thm:patterns-suff}}
\label{sec:proofthm4}
This section proves \Cref{thm:safe-backref} for each of Patterns~1–3.

\subsection{Proof for \Cref{pat:1}}
\begin{proof}
Pattern 1 contains a path of the form
\[ \pi_\mathrm{prefix} \xrightarrow{\lcg_i} \pi_\mathrm{left} \pi_\mathrm{pump}^* \pi_\mathrm{right} \xrightarrow{\rcg_i} \pi_\mathrm{bridge} \xrightarrow[s_{\mathrm{ref}}]{\backslash i} \pi_\mathrm{suffix} \]

Let $s_\mathrm{prefix} = \mathcal{S}(\pi_\mathrm{prefix})$. Assume that there exists a string $s_\mathrm{ovlp}$ such that
$\mathcal{S}(\pi_\mathrm{left}) = s_\mathrm{ovlp}^{u_l}$,
$\mathcal{S}(\pi_\mathrm{pump}) = s_\mathrm{ovlp}^{u_p}$,
$\mathcal{S}(\pi_\mathrm{right}) = s_\mathrm{ovlp}^{u_r}$, and
$\mathcal{S}(\pi_\mathrm{bridge}) = s_\mathrm{ovlp}^{u_b}$.
In addition, assume there exists a string $s_\mathrm{nsuffix}$ such that
$\mathcal{S}(\pi_\mathrm{suffix}) \neq s_\mathrm{nsuffix}$.

For any $n' \in \mathbb{N}$, construct the input string
\[ s = s_\mathrm{prefix} \; s_\mathrm{ovlp}^{2 (u_l + n' u_p + u_r) +u_b} \; s_\mathrm{nsuffix} \]

During execution, the prefix path $\pi_\mathrm{prefix}$ first matches $s_\mathrm{prefix}$. The loop $\pi_\mathrm{pump}^*$ then greedily matches as many copies of $s_\mathrm{ovlp}$ as possible. During backtracking, the number of iterations of $\pi_\mathrm{pump}$ is reduced by $u_p$ at each step until it reaches zero. When $\pi_\mathrm{pump}^*$ matches between $n' u_p$ and $0$ copies of $s_\mathrm{ovlp}$, the bridge path $\pi_\mathrm{bridge}$ matches $s_\mathrm{ovlp}^{u_b}$, after which the backreference $\backslash i$ is evaluated against
\[ u_l + n' u_p + u_r, \; u_l + (n' - 1) u_p + u_r, \; \cdots, \; u_l + u_p + u_r, \; u_l + u_r \]
copies of $s_\mathrm{ovlp}$. In all cases, the suffix path $\pi_\mathrm{suffix}$ rejects on $s_\mathrm{nsuffix}$, forcing continued backtracking.

As a result, the backreference is evaluated $n'$ times. The total time spent evaluating the backreference is
\begin{align*}
    &|s_\mathrm{ovlp}| \sum_{k=0}^{n'} (u_l + ku_p + u_r) \\
    &= |s_\mathrm{ovlp}| (n' + 1) (n' u_p / 2 + u_l + u_r),
\end{align*}
which is $\Omega(n'^2)$.
Since the total input length is
\[ |s| = |s_\mathrm{prefix}| + (2 (u_l + n' u_p + u_r) + u_b) |s_\mathrm{ovlp}| + |s_\mathrm{nsuffix}|, \]
it follows that $n' \in \Theta(|s|)$. Therefore, the time spent evaluating the backreference is $\Omega(|s|^2)$, and the overall matching runtime is not in $\mathcal{O}(|s|)$.

\end{proof}

\subsection{Proof for \Cref{pat:2}}
\begin{proof}
Pattern 2 consists of a path of the form
\[ \pi_\mathrm{prefix} \xrightarrow{\lcg_i} \pi_\mathrm{left} \pi_\mathrm{pump} \pi_\mathrm{right} \xrightarrow{\rcg_i} \pi_\mathrm{fence} \pi_\mathrm{loop}^* \pi_\mathrm{bridge} \xrightarrow[s_{\mathrm{ref}}]{\backslash i} \pi_\mathrm{suffix} \]

Let $s_\mathrm{prefix} = \mathcal{S}(\pi_\mathrm{prefix})$. Assume that there exists a string $s_\mathrm{ovlp}$ such that
$\mathcal{S}(\pi_\mathrm{left}) = s_\mathrm{ovlp}^{u_l}$,
$\mathcal{S}(\pi_\mathrm{pump}) = s_\mathrm{ovlp}^{n'1}$,
$\mathcal{S}(\pi_\mathrm{loop}) = s_\mathrm{ovlp}^{u_o}$, and
$\mathcal{S}(\pi_\mathrm{bridge}) = s_\mathrm{ovlp}^{u_b}$.
Let $s_\mathrm{right} = \mathcal{S}(\pi_\mathrm{right})$ and
$s_\mathrm{fence} = \mathcal{S}(\pi_\mathrm{fence})$.
In addition, assume that there exists a string $s_\mathrm{nsuffix}$ such that
$s_\mathrm{nsuffix} \neq \mathcal{S}(\pi_\mathrm{right}\pi_\mathrm{suffix})$.

Construct the input string
    \[
        s = s_\mathrm{prefix} \; s_\mathrm{ovlp}^{u_l + n'_1} \; s_\mathrm{right} \; s_\mathrm{fence} \; s_\mathrm{ovlp}^{n'_2 u_o + u_b + u_l + n'_1} \; s_\mathrm{nsuffix},
    \]
where $n'_1 \cdot n'_2 \notin \mathcal{O}(|s|)$ (for example, $n'_1, n'_2 \in \Theta(|s|)$).

During execution, the prefix path $\pi_\mathrm{prefix}$ matches $s_\mathrm{prefix}$, after which the capture group matches the substring $s_\mathrm{ovlp}^{u_l + n'1} \; s_\mathrm{right}$. The fence path $\pi_\mathrm{fence}$ then matches $s_\mathrm{fence}$. Subsequently, the loop $\pi_\mathrm{loop}^*$ greedily matches as many copies of $s_\mathrm{ovlp}$ as possible. During backtracking, the number of matched copies is reduced by $u_o$ at each step until it reaches zero.

When $\pi_\mathrm{loop}^*$ matches between $n'2 u_o$ and $0$ copies of $s_\mathrm{ovlp}$, the bridge path $\pi_\mathrm{bridge}$ matches $s_\mathrm{ovlp}^{u_b}$. The backreference $\backslash i$ then attempts to match the prefix of the previously captured string, namely $s_\mathrm{ovlp}^{u_l + n'1}$. In all cases, the remaining suffix and the path $\pi\mathrm{suffix}$ reject on $s_\mathrm{nsuffix}$, forcing further backtracking.

As a result, the backreference $\backslash i$ is evaluated $\Theta(n'_2)$ times, and each evaluation incurs a cost of $\Omega(u_l + n'_1)$. Consequently, the total time spent evaluating the backreference is $\Omega(n'_1 n'_2)$. Since $n'_1 n'_2 \notin \mathcal{O}(|s|)$, the overall matching runtime is not in $\mathcal{O}(|s|)$.

\end{proof}

\subsection{Proof for \Cref{pat:3}}
The proof for Pattern 3 is analogous to that for Pattern 2.

    %The proof for Pattern 3 is very similar to that of pattern 2. We leave it to readers.

\section{Snort REwB ReDoS Exploits}
\label{sec:snort}

\begin{exploit}{}{1}
\begin{exploitdesc}
    \item[Rules] SID 20156 Review 11, SID 20494 Review 19
    \item[Files] {\small\ttfamily snapshot-29200/rules/file-pdf.rules}, {\small\ttfamily snapshot-29200/rules/file-identify.rules}
    \item[PCRE Regex] ~
    \begin{exploitlst}[breaklines, breakident=0em, basicstyle=\ttfamily]
([A-Z\d_]+)\.write\x28.*?\1\.getCosObj\x28
    \end{exploitlst}
    \item[Attack String] ~
    \begin{exploitlst}
.write(.getCosObj(%PDF-Z.write(Z
    \end{exploitlst}
    (Repeat the 1st `\texttt{Z}' 1000 times, the 2nd one 2000 times.)
    \item[Effect] Slowing down 0.7-1.2 seconds.
    \item[Explanation] The substring ``\texttt{.write(.getCosObj(}'' satisfies the \texttt{content} constraint of the rule with SID 20156. The substring ``\texttt{\%PDF-}'' triggers the rule with SID 20494, causing the \texttt{file.pdf} flowbit to be set.  
    When the regex attempts to match the remaining portion of the input, it incurs $\mathcal{O}(n^2)$ time per match due to Pattern~2. Because the regex is not anchored, the PCRE engine attempts to start matching at multiple input positions. Each attempt that begins at a \texttt{Z} character in the first cluster results in an $\mathcal{O}(n^2)$ match. Consequently, the overall time complexity becomes $\mathcal{O}(n^3)$.
\end{exploitdesc}
\end{exploit}
    
\begin{exploit}{}{2}
\begin{exploitdesc}
    \item[Rules] SID 21081 Review 9
    \item[Files] {\small\ttfamily snapshot-29200/rules/deleted.rules}
    \item[PCRE Regex] ~
    \begin{exploitlst}
(\w+)\s*?\x3D\s*?document\x2Ecreateelement.*?\1\x2EsetAttribute.*?BD96C556-65A3-11D0-983A-00C04FC29E36.*?\1\x2EcreateObject\x28[\x22\x27]Shell\x2EApplication
    \end{exploitlst}
    \item[Attack String] ~
    \begin{exploitlst}
Shell.ApplicationZ=document.createelementZ.setAttributeBD96C556-65A3-11D0-983A-00C04FC29E36
    \end{exploitlst}
    (Repeat the 1st `\texttt{Z}' 1000 times, the 2nd one 2000 times.)
    \item[Effect] Slowing down 0.6-1.1 seconds.
    \item[Explanation] The strings \texttt{Shell.Application}'' and \texttt{setAttributeBD96}...'' are used to satisfy the \texttt{content} requirement. Capture group~1 \verb|/(\w+)/| together with the backreference \verb|/.?\1/| forms Pattern~2. As a result, when the regex attempts to match the remaining portion of the input, each matching attempt incurs $\mathcal{O}(n^2)$ time. Moreover, because the regex is not anchored, the PCRE engine repeatedly attempts to start matching at different input positions. Consequently, the overall time complexity becomes $\mathcal{O}(n^3)$.
\end{exploitdesc}
\end{exploit}

\begin{exploit}{}{3}
\begin{exploitdesc}
    \item[Rules] SID 10417 Review 10
    \item[Files] {\footnotesize\ttfamily snapshot-29200/rules/browser-plugins.rules}
    \item[PCRE Regex] ~
    \begin{exploitlst}
(\w+)\s*=\s*(\x22JNILOADER\.JNILoaderCtrl\x22|\x27JNILOADER\.JNILoaderCtrl\x27)\s*\x3b.*(\w+)\s*=\s*new\s*ActiveXObject\s*\(\s*\1\s*\)(\s*\.\s*(LoadLibrary)\s*\(|.*\3\s*\.\s*(LoadLibrary)\s*\()|(\w+)\s*=\s*new\s*ActiveXObject\s*\(\s*(\x22JNILOADER\.JNILoaderCtrl\x22|\x27JNILOADER\.JNILoaderCtrl\x27)\s*\)(\s*\.\s*(LoadLibrary)\s*\(|.*\7\s*\.\s*(LoadLibrary)\s*\()
    \end{exploitlst}
    \item[Attack String] ~
    \begin{exploitlst}
A='JNILOADER.JNILoaderCtrl';Z=new ActiveXObject(A);Z.LoadLibrary('org.evil.Malicious');
    \end{exploitlst}
    (Repeat both `\texttt{Z}'s 2000 times.)
    \item[Effect] Exceeds the backtracking limit, thereby bypassing alert generation.
    \item[Explanation] 
    
    Capture group~3 and its corresponding backreference place this regex in Pattern~2. In addition, the combination of capture group~3 with the preceding \verb|/.*(\w+)/| introduces an IDA pattern. When the regex is matched against the attack string, the greedy \verb|/.*/| initially consumes all occurrences of \texttt{Z} because it appears before \verb|/(\w+)/|. The engine must then repeatedly backtrack until \verb|(\w+)| can match the entire sequence of \texttt{Z} characters, allowing the backreference \verb|/\3/| to match. In practice, PCRE exceeds its backtracking limit before this state is reached, causing the match attempt to abort.

    Additionally, the input string is a snippet of malicious JavaScript code that can load arbitrary Java classes via ActiveX, illustrating a realistic exploitation scenario.

%    The capture group 3 and the backreference to it makes this regex fall in pattern 2. Additionally, note the capture group 3 and the \verb|/.*/| before it: \verb|/.*(\w+)/| makes this regex IDA. When matching the regex against the attacking string, \verb|/.*/| will try to match all ``\texttt{Z}'' since it is greedy and in front of \verb|/(\w+)|. However, it need to keep backtrack until \verb|(\w+)| can match all ``\textbb{Z}'', so that the \verb|/\3/| can be matched. However, PCRE will exceed its backtracking limit before that can be archived. Additionally, the input string is a piece of vulnerable JavaScript code, which can load arbitry Java class with the assist of ActiveX. 
\end{exploitdesc}
\end{exploit}

\begin{exploit}{}{4}
\begin{exploitdesc}
    \item[Rules] SID 51184 Review 2
    \item[Files] {\footnotesize\ttfamily snapshot-29200/rules/server-webapp.rules}
    \item[PCRE Regex:] ~
    \begin{exploitlst}
xmlns:(\S+)=[\x27\x22]http:\/\/xml\.apache\.org\/(xalan|xslt)[\x27\x22].*\1:(entities|content-handler)=([\x27\x22]((http|ftp).*?|(\S+\$\S+))[\x27\x22])
    \end{exploitlst}
    \item[Attack String] ~
    \begin{exploitlst}
<!--xmlns:B="http://xml.apache.org/xalan"B:entities="$--><xsl:output xmlns:xalan="http://xml.apache.org/xalan" xalan:entities="http://evil.org/malicious.bin"/>
    \end{exploitlst}
    (Repeat the `\texttt{\$}' 2000 times.)
    \item[Effect] Exceeding matching limit.
    \item[Explanation] This regex falls into Pattern~2. In addition, the subexpression \verb|/\S+\$\S+/| introduces an IDA pattern. The first portion of the attack input, enclosed by \texttt{<!--}'' and \texttt{-->}'', is an XML comment. Matching the regex against this portion causes PCRE to exceed its backtracking limit. The second portion is a valid XML node, which may allow Xalan-Java to load an arbitrary class.
\end{exploitdesc}
\end{exploit}

\end{document}